\documentclass[11pt,reqno,twoside]{article}

\usepackage{fixltx2e} 

\usepackage{cmap} 

\usepackage[T1]{fontenc}
\usepackage[utf8]{inputenc}
\usepackage{graphicx}
\usepackage{placeins}
\usepackage{enumerate}
\usepackage{algcompatible}

\usepackage{subcaption}

\usepackage{verbatim}

\newcommand{\gpa}{\mbox{ GPa}}

\usepackage{setspace}

\let\counterwithin\relax  
\usepackage{lmodern} 
\usepackage[scale=0.88]{tgheros} 

\usepackage{bm} 

\usepackage{bbold}

\usepackage{amsmath,amsbsy,amsgen,amscd,amsthm,amsfonts,amssymb} 

\usepackage[centering,top=1.in,bottom=1in,left=1in,right=1in]{geometry}

\usepackage{titling}
\usepackage{musicography}
\graphicspath{ {./images/} }

\usepackage[sf,bf,compact]{titlesec}

\usepackage{booktabs,longtable,tabu} 
\usepackage[font=small,margin=12pt,labelfont={sf,bf},labelsep={space}]{caption}

\usepackage[usenames,dvipsnames]{xcolor}

\definecolor{dark-gray}{gray}{0.3}
\definecolor{dkgray}{rgb}{.4,.4,.4}
\definecolor{dkblue}{rgb}{0,0,.5}
\definecolor{medblue}{rgb}{0,0,.75}
\definecolor{rust}{rgb}{0.5,0.1,0.1}

\usepackage{url}
\usepackage[colorlinks=true]{hyperref}
\hypersetup{linkcolor=dkblue}    
\hypersetup{citecolor=rust}      
\hypersetup{urlcolor=rust}     

\usepackage[final]{microtype} 

\newtheoremstyle{myThm} 
    {\topsep}                    
    {\topsep}                    
    {\itshape}                   
    {}                           
    {\sffamily\bfseries}                   
    {.}                          
    {.5em}                       
    {}  

\newtheoremstyle{myRem} 
    {\topsep}                    
    {\topsep}                    
    {}                   
    {}                           
    {\sffamily}                   
    {.}                          
    {.5em}                       
    {}  

\newtheoremstyle{myDef} 
    {\topsep}                    
    {\topsep}                    
    {}                   
    {}                           
    {\sffamily\bfseries}                   
    {.}                          
    {.5em}                       
    {}  

\theoremstyle{myThm}
\newtheorem{theorem}{Theorem}[section]

\newtheorem{proposition}[theorem]{Proposition}

\theoremstyle{myRem}
\newtheorem{remark}[theorem]{Remark}

\theoremstyle{myDef}

 \newtheorem{example}[theorem]{Example}

\usepackage{fancyhdr}
\usepackage{nopageno} 
\usepackage{algorithm}
\usepackage{algpseudocode}

\let\originalleft\left
\let\originalright\right
\renewcommand{\left}{\mathopen{}\mathclose\bgroup\originalleft}
\renewcommand{\right}{\aftergroup\egroup\originalright}

\usepackage{mathtools}
\mathtoolsset{centercolon}  

\renewcommand{\phi}{\varphi}

\providecommand{\mathbbm}{\mathbb} 

\newcommand{\R}{\mathbbm{R}}

\newcommand{\N}{\mathbbm{N}}

\renewcommand{\O}{\mathcal{O}}

\newcommand{\K}{\mathcal{K}}
\newcommand{\U}{\mathcal{U}}
\newcommand{\UU}{U'}
\renewcommand{\AA}{A'}
\newcommand{\YY}{Y'}

\newcommand{\Z}{\mathcal{Z}}
\newcommand{\A}{\mathcal{A}}
\newcommand{\Y}{\mathcal{Y}}

\newcommand{\G}{\mathcal{G}}
\newcommand{\F}{\mathcal{F}}

\definecolor{mygreen}{rgb}{0.1,0.75,0.2}

\newcommand{\nc}{\normalcolor}

\newcommand{\z}{z}

\newcommand{\md}{\mbox{d}}

\usepackage[]{algorithm}

\usepackage{graphicx}

\usepackage{authblk}
\usepackage[square,numbers]{natbib}
\usepackage{chngcntr}
\usepackage{mathrsfs}
\counterwithin{table}{section}
\counterwithin{algorithm}{section}
\usepackage{soul}
\usepackage{authblk}

\title{Data-Driven Forward Discretizations for Bayesian Inversion} 
\author[1]{D. Bigoni}
\author[2]{Y. Chen}
\author[3]{N. Garcia Trillos}
\author[1]{Y. Marzouk}
\author[2]{D. Sanz-Alonso}

\affil[1]{Massachusetts Institute of Technology}
\affil[2]{University of Chicago}
\affil[3]{University of Wisconsin, Madison}
\vspace{.25in} 

\date{}

\makeatletter\@addtoreset{section}{part}\makeatother%
\numberwithin{equation}{section}

\newcommand{\upperRomannumeral}[1]{\uppercase\expandafter{\romannumeral#1}}

\begin{document}
\maketitle 


\begin{abstract}
This paper suggests a framework for the learning of discretizations of expensive forward models in  Bayesian inverse problems. The main idea is to incorporate the parameters governing the discretization as part of the unknown to be estimated within the Bayesian machinery. We numerically show that in a variety of inverse problems arising in mechanical engineering, signal processing and the geosciences, the observations contain useful information to guide the choice of discretization.
\end{abstract}

	\section{Introduction}\label{sec:introduction}
	Models used in science and engineering are often described by problem-specific input parameters that are estimated from indirect and noisy observations. The inverse problem of input reconstruction is defined in terms of a \emph{forward model} from inputs to observable quantities, which in many applications needs to be approximated by discretization. A broad class of examples motivating this paper is the reconstruction of input parameters of differential equations. The choice of forward model discretization is particularly important in Bayesian formulations of inverse problems: discretizations need to be cheap since statistical recovery may involve  millions of evaluations of the discretized forward model; they also need to be accurate enough to enable input reconstruction. The goal of this paper is to suggest a simple data-driven framework to build forward model discretizations to be used in Bayesian inverse problems. The resulting discretizations are data-driven in that they  finely resolve  regions of the input space where the data are most informative, while keeping the cost moderate by  coarsely resolving  regions that are not informed by the data. 
	
	To be concrete and explain the idea, let us consider the inverse problem of recovering an unknown $u$ from data $y$ related by
	\begin{equation}\label{eq:toydescription}
	y = \G(u) + \eta,
	\end{equation}
	where $\G$ denotes the forward model from inputs to observables, $\eta\sim N(0,\Gamma)$ represents model error and observation noise, and $\Gamma$ denotes a positive definite noise covariance matrix. We will follow a Bayesian approach, viewing $u$ as a random variable \cite{kaipio2006statistical,AS10,sanz2018inverse} with  \emph{prior} distribution $p_u(u).$  The Bayesian solution to the inverse problem is the \emph{posterior} distribution $p_{u|y}(u)$ of $u$ given the data $y,$ which by an informal application of Bayes theorem is characterized by
	\begin{align}\label{eq:posteriorintro}
	\begin{split}
		 p_{u|y}(u)  \propto \exp \bigl(-\Phi(u;y)\bigr)p_u(u), \quad \quad  \Phi(u;y):= \frac12 \left\| y - \G(u) \right\|_\Gamma^2
	\end{split}
	\end{align}
	with $\| \cdot \|_\Gamma := \| \Gamma^{-1/2} \cdot \|.$
	A common computational bottleneck arises when the forward model $\G$ and hence the likelihood are intractable, meaning that it is impossible or too costly to evaluate. This paper introduces a framework to tackle this computational challenge by employing data-driven discretizations of the forward model. The main idea is to include the parameters that govern the discretization as part of the unknown to be estimated within the Bayesian machinery. More precisely, we consider a family  $\{\G^a\}_{a\in \A}$ of approximate forward models and put a prior $q_{u,a}(u, a)$ over both unknown inputs $u$ and forward discretization parameters $a\in \A$  to define a joint posterior
	\begin{align}\label{eq:posteriorsurrogate}
	\begin{split}
	q_{u,a|y}(u,a) \propto \exp \bigl(-\Psi(u,a) \bigr) q_{u,a}(u, a), \quad \quad \Psi(u, a;y):= \frac12 \left\| y - \G^a(u) \right\|_\Gamma^2.
	\end{split}
	\end{align}
	While this structure underlies many hierarchical formulations of Bayesian inverse problems \cite{kaipio2006statistical}, in this paper the hyper-parameter $a$  determines the choice of discretization of the forward model $\G.$
	
	Including the learning of the numerical discretizations of the forward map as part of the inference agrees with the Bayesian philosophy of treating unknown quantities as random variables, and is also in the spirit of recent probabilistic numerical methods \cite{cockayne2019bayesian}; rather than implicitly assuming that a true hidden numerical discretization of the forward model generates the data, a Bayesian would acknowledge the uncertainty in the choice of a suitable discretization and let the observed data inform such a choice. Moreover, the Bayesian viewpoint has two main practical advantages. First, data-informed grids will typically be coarse in regions of the input space that are not informed by the data, allowing successful input reconstruction at a reduced computational cost. Second, the posterior $q_{u,a|y}(u,a)$ contains useful uncertainty quantification on the discretizations. This additional uncertainty information may be exploited to build a high-fidelity forward model to be employed within existing inverse problem solvers, either in Bayesian or classical settings. 
	
	\subsection{Related Work}\label{ssec:literature}
	The Bayesian formulation of inverse problems provides a flexible and principled way to combine data with prior knowledge. However, in practice it is rarely possible to perform posterior inference with the model of interest  \eqref{eq:posteriorintro} due to various computational challenges. In this paper we investigate the construction of computable data-driven forward discretizations of intractable likelihoods arising in the inversion of differential equations. Other intertwined obstacles for posterior inference are:
	\begin{itemize}
		\item {\bf Sampling cost.} While exact posterior inference is often intractable, approximate posterior inference can be performed by employing sampling algorithms. Markov chain Monte Carlo and particle-based methods are popular, but  implementations of these algorithms require repeated evaluation of the forward model $\G,$ which may be costly.
		\item {\bf Large input dimension.} The unknown parameter $u$ may be high, or even infinite dimensional. While the convergence rate of certain sampling schemes may be independent of the input dimension \cite{cotter2013mcmc, agapiou2015importance, trillos2017consistency}, the computational and memory cost per sample may increase prohibitively with dimension. 
		\item {\bf Model error.} The forward model is only an approximation of the real relationship between input and observable output variables. Model discrepancy can damage input recovery. 
		\item {\bf Complex geometry.} The unknown may be a function defined on a complex, perhaps unknown domain that needs to be approximated. 
	\end{itemize}
	
	All these challenges have long been identified \cite{sacks1989design,kennedy2001bayesian,kaipio2006statistical,kaipio2007statistical},  giving rise to a host of methods for sampling, parameter reduction, model reduction, enhanced model error techniques  and  geometric methods for inverse problems. We focus on the model-reduction problem of building forward discretizations, but the methodology proposed in this paper can be naturally combined with existing techniques that address complementary challenges.  For instance, our forward model discretizations may be used within multilevel MCMC methods \cite{giles2008multilevel} or within two-stage sampling methods \cite{green2001delayed,tierney1999some,christen2005markov,cui2015data,efendiev2006preconditioning}, and thus help to reduce the sampling cost. Also, forward model discretizations may be combined with parameter reduction and model adaptation techniques, as in \cite{lieberman2010parameter,li2018model}. It is important, however, to distinguish between the parameter and model reduction problems. While the former aims to find suitable small-dimensional representations of the input $u$, the latter is concerned with effectively reducing the number of degrees of freedom used to compute the forward model $\G.$  In regards to model error, our framework may be thought of as incorporating Bayesian model choice to the Bayesian solution of inverse problems by viewing each forward model discretization as a potential model. Following this interpretation, the \emph{a posteriori}  choice of forward discretization may in principle be determined using Bayes factors. Lastly, learning appropriate discretizations of forward models is particularly important for inverse problems set in complex, possibly uncertain geometries \cite{trillos2017consistency,garcia2018continuum,harlim2019kernel}.
	
	Many approaches to computing forward map surrogates and reduced-order models have been proposed; we refer to \cite{frangos2010surrogate} for an extended survey, and to \cite{peherstorfer2018survey} for a broader discussion of multi-fidelity models in other outer-loop applications. Most methods fall naturally into one of three categories: 
	\begin{enumerate}
		\item Projection-based methods: the forward model equations are described in a reduced basis that is constructed using few high-fidelity forward solves (called snapshots). Two popular ways to construct the reduced basis are proper orthogonal decomposition (POD) and reduced order basis. In the inverse problem context, data-informed construction of snapshots  \cite{cui2015data} allows to approximate the posterior support with fewer high-fidelity forward runs. To our knowledge, there is little theory to guide the required number or location of snapshots to meet a given error tolerance. 
		\item Spectral methods: polynomial chaos \cite{xiu2002wiener} is a popular method for forward propagation of uncertainty, that has more recently been used to produce surrogates for intractable likelihoods \cite{marzouk2007stochastic}. The paper \cite{marzouk2009stochastic} translates error in the likelihood approximation to Kullback-Leibler posterior error. A drawback of these methods is that they are only practical when the random inputs can be represented by a small number of random variables. 
		Recent interest lies in adapting the spectral approximations to observed data \cite{li2014adaptive}.
		\item Gaussian processes and neural networks: some of the earliest efforts to allow for Bayesian inference with complex models suggested to use Gaussian processes  \cite{rasmussen2006gaussian} to construct surrogate likelihood models \cite{sacks1989design,kennedy2001bayesian}. The accuracy of the resulting approximations has been studied in \cite{stuart2017posterior}, which again requires a suitable representation of the input space. Finally, representation of the likelihood using neural networks in combination with generalized polynomial chaos expansions
		 has been investigated in   \cite{schwab2019deep}.
	\end{enumerate}
	This paper focuses on grid-based discretizations and density-based discretizations of static inverse problems arising in mechanical engineering,  signal processing and the geophysical sciences. However, the proposed framework may be used in conjunction with other reduced-order models, in dynamic data assimilation problems, and in other applications. Finally, we mention that for classical formulations of certain specific inverse problems, optimal forward discretization choices have been proposed \cite{borcea2005continuum,becker2005mesh}. 
	
	\subsection{Outline and Contributions}
	Section \ref{sec:fullproblem} reviews the Bayesian formulation of inverse problems. Section \ref{sec:bayesiandiscretization} describes the main framework for the Bayesian learning of forward map discretizations. We will consider two ways to parametrize  discretizations: in the first, the grid points locations are learned directly, and in the second we learn a probability density from which to obtain the grid. In Section \ref{sec:sampling} we discuss a general approach to sampling the joint posterior over unknown input and discretization parameters, which consists of a Metropolis-within-Gibbs that alternates between a reversible jump Markov chain Monte Carlo (MCMC) algorithm to update the discretization parameters and a standard MCMC to update the unknown input. Section \ref{sec:numerics} demonstrates the applicability, benefits, and limitations of our approach in a variety of inverse problems arising in mechanical engineering, signal processing and source detection,  considering Euler discretization of ODEs, Euler-Maruyama discretization of SDEs, and finite element methods for PDEs. We conclude in Section \ref{sec:conclusions} with some open questions for further research.
	
	\section{Background: Bayesian Formulation of Inverse Problems}\label{sec:fullproblem}
	Consider the inverse problem of recovering an unknown $u \in \U$ from data $y\in\R^{m}$ related by
	\begin{equation}
	\label{eq:inverseproblem}
	y = \G(u) + \eta,
	\end{equation}
	where $\U$ is a space of admissible unknowns and $\eta$ is a random variable whose distribution is known to us, but not its realization. In many applications, the \emph{forward model} $\G:\U \to \R^{m}$  can be written as the composition of forward and observation maps, $\G = \O\circ \F$. The forward map $\F: \U \to \Z$ represents a \emph{complex} mathematical model that assigns outputs $\z \in \Z$ to inputs $u \in \U$. For instance,  $u$ may be the parameters of a differential equation, and $z$ may be its analytical solution. The observation map $\O: \Z \to \Y$ establishes a link between outputs and observable quantities, e.g. by point-wise evaluation of the solution. 
	
In the Bayesian formulation of the inverse problem \eqref{eq:inverseproblem}, one specifies a \emph{prior} distribution on $u$ and seeks to characterize the \emph{posterior} distribution of $u$ given $y.$ If the input space $\U$ is finite dimensional,  $\U\subset\R^{d}$, then the prior distribution, denoted as $p_u(u),$ can be defined through its Lebesgue density. The noise distribution of $\eta$ in $\R^{m}$ gives the \emph{likelihood} $p_{y|u}(y|u)$. In this work we assume, for concreteness, that $\eta$ is a zero-mean Gaussian with covariance $\Gamma \in \R^{m \times m}$, so that 
	\begin{equation}\label{eq:likelihood}
	p_{y|u}(y| u) \propto \exp \bigl(-\Phi(u;y)\bigr), \quad \quad  \Phi(u;y):= \frac12 \left\| y - \G(u) \right\|_\Gamma^2,
	\end{equation} 
	where $\| \cdot \|_\Gamma := \| \Gamma^{-1/2} \cdot \|$. Using Bayes' formula, the posterior density is given by
	\begin{equation}\label{eq:posteriorlebesgue}
	p_{u|y}(u) =\frac{1}{Z}\, p_{y|u}(y|u) p_u (u), \quad \quad Z=\int_\U p_{y|u}(y|u)p_u(u)\md u
	\end{equation}
	with multiplicative constant $Z$ depending on $y$.

For many inverse problems of interest, the unknown $u$ is a function and  the input space $\U$ is an infinite-dimensional Banach space. In such a case, the prior cannot be specified in terms of its Lebesgue density, but rather as a measure $\mu_u$ supported on $\U.$ Provided that $\G: \U \to \R^{m}$  is measurable and that $\mu_u(\U) = 1, $ the posterior measure $\mu_{u|y}$ is still defined,  in analogy to \eqref{eq:posteriorlebesgue}, as a change of measure with respect to the prior
	\begin{equation}\label{eq:posterior}
	\frac{\md \mu_{u|y}}{\md \mu_u}(u)\propto \exp\bigl(-\Phi(u;y)\bigr).
	\end{equation}
 We refer to  \cite{AS10} and \cite{trillos2016bayesian} for further details. The posterior $\mu_{u|y}$  contains, in a precise sense \cite{duke}, all the information on $u$ available in the data $y$ and the prior $\mu_u$. This paper is concerned with inverse problems where  $\G=\O\circ\F$ arises from a complex model $\F$ that cannot be evaluated pointwise; we then seek to approximate the \emph{idealized} posterior $\mu_{u|y}$ finding a compromise between accuracy and computational cost. 
 
 A simple but important observation is that approximating $\F$ accurately is \emph{not} necessary in order to approximate $\mu_{u|y}$ accurately. It is \emph{only} necessary to approximate $\G= \O \circ \F,$ since $\F$ appears in the posterior only through $\G.$ While producing discretizations to complex models $\F$ has been widely studied in numerical analysis, here we investigate how to approximate $\F$ with the specific goal of approximating the posterior $\mu_{u|y},$ incorporating prior and data knowledge into the discretizations. For some inverse problems the observation operator $\O$ also needs to be discretized, leading to similar considerations. 
	

	\section{Bayesian Discretization of the Forward Model}\label{sec:bayesiandiscretization}
	Suppose that $\F$ is the solution map to a differential equation that cannot be solved in closed form, and $\O$ is point-wise evaluation of the solution. Standard practice in computing the Bayesian solution to the inverse problem involves using an \emph{a priori}  fixed discretization, e.g., by discretizing the domain of the differential equation into a fine grid. Provided that the grid is fine enough, the posterior  defined with the discretized forward map can approximate well the one in  \eqref{eq:posterior}. However, the discretizations are usually performed on a  \emph{fine} uniform grid  which may lead to unnecessary waste of computational resources. Indeed, it is expected that the choice of discretization should be problem dependent, and should be informed both by the observation locations (which are often not uniform in space) and by the value of the unknown input parameter that we seek to reconstruct. Thus we seek to learn  \emph{jointly} the unknown input $u$ and the discretization of the forward map.

We will consider a parametric family of discretizations. Precisely, we let
	\begin{equation}\label{eq:surrogatespace}
	\A:= \Bigl\{ a = (k, \theta): \, k \in \K \subset\{1,2, \ldots\},\, \theta \in D(k) \subset \R^{d(k)}\Bigr\},
	\end{equation}
	and each pair $a=(k,\theta)\in \A$ will parameterize a discretized forward model $\G^a.$ 
	For given $k\in\K$, $d(k)$ represents the degrees of freedom in the discretization, and $\theta\in D(k)$ is the $d(k)$-dimensional model parameter of the discretization, where $D(k)$ is the region containing all parameters of interest. In analogy with Bayesian model selection frameworks \cite{robert2007bayesian,green1995reversible}, $k \in \K$ may be interpreted as indexing the discretization model.	We focus on the model reduction rather than the parameter reduction problem, and assume that all approximation maps share the same input and output spaces $\U$ and $\Z.$  We will illustrate the flexibility of this framework using grid-based approximations and density-based discretizations.
	
	\begin{example}[Grid-based discretizations]\label{ex:gridbased}
	\label{1}
	Here the first component of each element $a = (k,\theta) \in \A$ represents the number of points in a grid. The set $\K$ contains all allowed grid sizes. If we denote $D\subset \mathbb{R}^d$ as the temporal or spatial domain of the equation being discretized, we define $D(k) = D^k$ and $d(k) = d\times k$, where $D^k := D \times \cdots  \times D$ denotes the $k$-fold Cartesian product of $D$. Then the second component $\theta = [x_1, \ldots x_k]$ encodes the locations of $k$ grid points.
	\end{example}
	
	\begin{example}[Density-based discretizations]\label{ex:densitybased}
	 Here the first component of each element $a = (k,\theta ) \in \A$ represents again the number of points in a grid, and the second component parametrizes a probability density $\rho=\rho(x;\theta)$ on the temporal or spatial domain of interest, by a parameter $\theta$ of fixed dimension, independent of $k$. Given $a \in \A$ we may for instance employ MacQueen's method \cite{du2002} to formulate a centroidal Voronoi tessellation, which outputs $k$ generators $\{x_1,\ldots,x_k\}$, and then use them as grid points to generate a finite element grid by Delaunay triangulation. Intuitively $\theta$ controls the spatial  density  of the non-uniform grid points $\{x_1,\ldots,x_k\}$. The space $\K$  represents, as before, all the allowed number of grid points. 
	\end{example}
	
	\begin{example}[Other discretizations]
	As mentioned in the introduction, other discretizations and model reduction techniques could be considered within the above framework, including projection-based approximations, Gaussian processes, and graph-based methods. However, in our numerical experiments we will focus on grid-based and density-based discretizations.
	\end{example}

	We consider a product prior on $(u, a) \in \U \times \A,$ given by
	\begin{equation}\label{eq:prior}
	q_{u,a}(u,a) = q_u(u) q_a(a),
	\end{equation}
	where $q_u(u) = p_u(u)$ is as in the original, idealized inverse problem \eqref{eq:inverseproblem}.
	In general, conditioning on $u$ may or may not provide useful information about how to approximate $\G(u).$  When it does, this can be infused into the prior by letting the conditional distribution of $a$ given $u$  depend on $u.$ For simplicity we restrict ourselves to the product structure \eqref{eq:prior}.
	
	The examples above and the structure of the space $\A$ defined in equation \eqref{eq:surrogatespace} suggest to define hierarchically a prior over $a\in\A$
	\begin{equation}\label{eq:priorsurrogates}
	q_a(a) =q_{k,\theta}(k,\theta) = q_k(k) q_{\theta|k}(\theta|k),
	\end{equation}
	where $q_k(k)$ is a probability mass function that penalizes expensive discretizations that employ large number $d(k)$ of degrees of freedom, and $q_{\theta|k}(\theta|k)$ denotes the conditional distribution of $\theta$ given $k$ in $D(k).$

	We define the likelihood of observing data $y$ given $(u,a)$ by 
	\begin{equation}
	\label{eq:likelihoodApprox}
	q_{y|u,a}(y |u,a) \propto \exp \bigl(-\Psi(u,a;y) \bigr), \quad \quad \Psi(u, a;y):= \frac12 \left\| y - \G^a(u) \right\|_\Gamma^2,
	\end{equation}
	where $\G^a=\O\circ\F^a$. The discretized forward maps $\F^a$ will be chosen so that evaluating $\Psi$ is possible.
	
	We first consider the case where $\K=\{k\}$ is a singleton, and $\A:= \Bigl\{ a = (k, \theta): \theta \in D(k) \subset \R^{d(k)}\Bigr\}$ has a Euclidean space structure. Then, by Bayes' formula, 
	\begin{equation}\label{eq:posteriorApprox}
		q_{u,a|y}(u,a)=\frac{1}{\tilde{Z}}q_{y|u,a}(y|u,a)q_u(u)q_a(a),\quad\quad \tilde{Z}=\int_{\U\times\A}q_{y|u,a}(y|u,a)q_u(u)q_a(a)\md u \md a,
	\end{equation}
	where $\tilde{Z} = \tilde{Z}(y)$ is a normalizing constant.
	The first marginal of $q_{u,a|y}(u,a)$, which we denote as $q_{u|y}(u)$, constitutes a data-informed approximation of the posterior $p_{u|y}(u)$ from the full, idealized inverse problem \eqref{eq:posteriorlebesgue}. We have the following result:

	\begin{proposition}
		Let $p_{u|y}(u)$ be defined as in \eqref{eq:posteriorlebesgue} and $q_{u|y}(u)$ be defined as above. If $\G$ is bounded and, for $q_{u,a}$-almost any $(u,a)$, $\|\G^a(u)-\G(u)\|<\epsilon,$ \nc then 
		\begin{equation}
			d_{TV} \bigl(q_{u|y}(u), p_{u|y}(u)\bigr)<C\epsilon
		\end{equation}
		for some constant $C$ independent of $\epsilon$.
	\end{proposition}
	\begin{proof}
		Integrating both sides of the first equation in \eqref{eq:posteriorApprox} with respect to $a$:
		\[q_{u|y}(u)=\frac{1}{\tilde{Z}} \Bigl(\int_\A q_{y|u,a}(y|u,a)q_a(a) \md a\Bigr) q_u(u)=:\frac{1}{\tilde{Z}}\tilde{g}_y(u)q_u(u).\]
		Compare this to equation \eqref{eq:posteriorlebesgue} and write $g_y(u)=p_{y|u}(y|u)$, we have
		\begin{align*}
			\|\tilde{g}_y(u)-g_y(u)\|\le\int_\A \exp\left(-\frac{1}{2}\|y-\G^a(u)\|_\Gamma^2\right)-\exp\left(-\frac{1}{2}\|y-\G(u)\|_\Gamma^2\right)\md a\le C\|\G^a(u)-\G(u)\|,
		\end{align*}
		where the last inequality follows from the Lipschitz continuity of $e^{-w}$ for $w\ge 0$, boundedness of $\G$, and equivalence of norm in $\R^m$. This implies that $|\tilde{g}_y(u)-g_y(u)|\le C\epsilon$ and hence $|\tilde{Z}-Z|\le C\epsilon$. Then the statement follows from a slight modification of Theorem 1.14 in \cite{sanz2018inverse} and the definition of TV distance.
	\end{proof}

	Now we are ready to extend the above results to infinite-dimensional input space $\U$. We define a prior measure on $\U\times\A$ given by $\nu_{u,a}(du,da)=\nu_u(du)\times\nu_a(da)$, where $\nu_u(du) = \mu_u(du)$ is as in the idealized inverse problem. The posterior measure on $\U\times\A$ conditioning on $y$ will still be denoted by $\nu_{u,a|y}$.
	\begin{proposition}\label{prop:post}
	    Suppose that $\U$ is a separable Banach space with $\nu_u(\U)=1$, $\Psi: \U \times \A \to \R$ is continuous.
	    Then the posterior measure $\nu_{u,a|y}$ of $(u,a)$ given $y$ is absolutely continuous with respect to the prior $\nu_{u,a}$ on $\U\times\A$ and has Radon-Nikodym derivative
		\begin{equation}\label{eq:jointposterior}
		\frac{\md\nu_{u,a|y}}{\md\nu_{u,a}}(u,a) \propto \exp \bigl(-\Psi(u,a;y) \bigr).
		\end{equation} 
	\end{proposition}
	\begin{proof}
By the disintegration theorem (which holds for arbitrary Radon measures on separable metric spaces --see \cite{dellacherie2011probabilities} chapter 3, page 70) for all measurable subsets $\UU \subseteq \U$, $\AA\subseteq \A$ and $\YY \subseteq \Y$, we can write $\nu_{u,a,y}(\UU \times \AA \times \YY)$ in two different ways:
\[ \int_{\YY}\nu_{u,a|y}(\UU \times \AA|y) \md \nu_y(y)
=\nu_{u,a,y}(\UU, \AA, \YY)=\int_{\UU \times \AA} \nu_{y|u,a} (\YY | u,a) \md \nu_{u,a}(u,a).
   \]
In particular, 
\[ \nu_y( \YY ) = \int_{\U \times \A} \nu_{y|u,a}(\YY | u,a) \md \nu_{u,a}(u,a) = \int_{\U \times \A} \int_{\YY} Z_\Gamma^{-1}\exp\left(- \frac{1}{2} \lVert y-\G^a(u) \rVert_{\Gamma}^2 \right) \md y \md \nu_{u,a}(u,a), \]
given our assumptions on the noise model, where $Z_\Gamma$ is a constant depending on the noise covariance $\Gamma$. We can then use Tonelli's theorem to swap the order of the integrals and obtain
\[\nu_y(\YY) = \int_{\YY} \left( \int_{\U \times \A} Z_\Gamma^{-1}\exp\left(- \frac{1}{2} \lVert y-\G^a(u) \rVert_{\Gamma}^2 \right)d\nu_{u,a}(u,a) \right) \md y.\]
Given that $\YY$ is arbitrary, we conclude that $\nu_y$ is absolutely continuous with respect to the Lebesgue measure with density: 
\[ \frac{\md \nu_y(y)}{\md y}= \int_{\U \times \A} Z_\Gamma^{-1}\exp\left(- \frac{1}{2} \lVert y-\G^a(u) \rVert_{\Gamma}^2 \right)\md \nu_{u,a}(u,a). \]
On the other hand,
\begin{align*}
\int_{\UU \times \AA} \nu_{y|u,a} (\YY | u,a) \md \nu_{u,a}(u,a) & = \int_{\UU \times \AA} \int_{\YY} Z_\Gamma^{-1} \exp\left(- \frac{1}{2}\lVert y-\G^a(u)\rVert_\Gamma^2 \right) \md y \md \nu_{u,a}(u,a)
\\ &=\int_{\UU \times \AA} \int_{\YY} Z_\Gamma^{-1} \exp\left(- \frac{1}{2}\lVert y-\G^a(u)\rVert_\Gamma^2 \right)\left(\frac{\md \nu_y(y)}{\md y}\right)^{-1} \md \nu_y(y) \md \nu_{u,a}(u,a)
\\ &= \int_{\YY} \left(\int_{\UU \times \AA} Z_\Gamma^{-1} \exp\left(- \frac{1}{2}\lVert y-\G^a(u)\rVert_\Gamma^2 \right)\left(\frac{\md \nu_y(y)}{\md y}\right)^{-1} \md \nu_{u,a}(u,a)\right) \md \nu_y(y),
\end{align*}
applying Tonelli's theorem once again to obtain the last equality. Since $\YY$ was arbitrary, it follows that for $\nu_y$-a.e. $y$ we have
\[ \nu_{u,a|y}(\UU \times \AA|y) = \int_{\UU \times \AA} Z_\Gamma^{-1} \exp\left(- \frac{1}{2}\lVert y-\G^a(u)\rVert_\Gamma^2 \right)\left(\frac{\md \nu_y(y)}{\md y}\right)^{-1} \md \nu_{u,a}(u,a).\]
In turn, from the arbitrariness of $\UU, \AA$ it follows that, for $\nu_y$-a.e. $y$ the measure $\nu_{u,a|y}(\cdot| y)$ is absolutely continuous with respect to $\nu_{u,a}$ and its Radon-Nykodym derivative satisfies
\[ \frac{\md\nu_{u,a|y}}{\md\nu_{u,a}}(u,a) \propto \exp\bigl(- \frac{1}{2} \lVert y-\G^a(u)\rVert_{\Gamma}^2 \bigr) \]
as claimed, where the constant of proportionality depends on $y$.

	\end{proof}
	
As in the finite dimensional case, we have the following result:
	
	\begin{proposition}[Well-posedness of Posterior]
	Under the same assumption as in Proposition \ref{prop:post}, suppose further that for $q_{u,a}$-almost any $(u,a)$, $\|\G^a(u)-\G(u)\|<\epsilon$, and $\G$ is bounded. Then we have
	\begin{equation}
	    d_{TV}(\nu_{u|y}, \mu_{u|y})<C\epsilon
	\end{equation}
	for some constant $C$ independent of $\epsilon.$
	\end{proposition}
	\begin{remark}
	In the context of grid-based forward approximations, the condition `$\|\G^a(u)-\G(u)\|<\epsilon$ $q_{u,a}$-almost surely' can be interpreted as `almost any draw from the approximation parameter space $\A$ can produce an approximation of the forward model with error at most $\epsilon$'. This is often the case, for example, when the grids are finer than some threshold under regularity conditions on the input space.
	\end{remark}

	\section{Sampling the Posterior}\label{sec:sampling}
	The structure of the joint posterior $\nu_{u,a|y}$ over unknowns $u\in \U$ and approximations $a\in \A$ suggests using a Metropolis-within-Gibbs sampler, which constructs a Markov chain $(u^{(n)}, a^{(n)})$ by alternatingly  sampling each coordinate:
	
	\FloatBarrier
	\begin{algorithm}
		\caption{Metropolis-within-Gibbs Core Structure}
		\label{algorithm:gibbs}
		\begin{algorithmic}
			\State Choose $(u^{(1)}, a^{(1)}) \in \U \times \A$.
			\For{$n = 1:N$ } 
			\State  1. Sample $u^{(n+1)} \sim \mathbb{K}^{a^{(n)},y} (u^{(n)}| \cdotp).$
			\State 2. Sample  $a^{(n+1)} \sim \mathbb{L}^{u^{(n+1)},y} (a^{(n)}| \cdotp).$
			\EndFor
		\end{algorithmic}
	\end{algorithm}
	\FloatBarrier
	In the above, $\mathbb{K}^{a,y}$ and $\mathbb{L}^{u,y}$ are Metropolis-Hastings Markov kernels that are reversible with respect to $u| (a, y)$ and $a|(u, y).$ We remark that the kernel $\mathbb{K}^{a,y}$ involves evaluation of the forward model approximation $\G^a$ but not of the intractable full model $\G.$ 
	While the choice and design of the kernels $\mathbb{K}^{a,y}$ and $\mathbb{L}^{u,y}$ will clearly be problem-specific, and here we consider a standard method appropriate for the case where the input space $\U$ is a space of functions  to define $\mathbb{K}^{a,y}.$
	
	Before describing how to sample the full conditionals $\nu_{u|a,y}$ and $\nu_{a|u,y}$ of $u| (a, y)$ and $a |(u, y)$  it is useful to note that they satisfy the following expressions:
		\begin{equation}
		\frac{\md \nu_{u|a,y}}{\md \nu_u}(u) \propto \exp \Bigl(-\Psi(u,a;y) \Bigr), \quad  \frac{\md \nu_{a|u,y}}{\md \nu_a}(a) \propto  \exp \Bigl(-\Psi(u,a;y) \Bigr).
		\end{equation}

	\subsection{Sampling the Full Conditional  $u|y,a$}
	For  given $a$ and $y$, we can sample from  $\nu_{u|a,y}$ using pCN \cite{beskos2008mcmc}, with proposal
	\[ \tilde{u} : = \sqrt{1- \beta^2} u + \beta \xi, \quad \quad \xi \sim \mu_u ,\]
	and acceptance probability
	\[ \alpha(u , \tilde{u}) := \min \Bigl\{1, \exp\bigl(- \Psi(\tilde{u},a; y ) +  \Psi(u,a; y)  \bigr) \Bigr\}. \]
	 Other discretization-invariant MCMC samplers \cite{cui2016dimension,rudolf2015generalization} could also be used to update $u \vert y, a$, but pCN is a straightforward and effective choice in the examples considered here.

\subsection{Sampling the Full Conditional  $a |y,u$}
\subsubsection{Sampling Grid-based Discretizations}
We will use a Markov kernel $\mathbb{L}^{u,y}(a|\cdot)$ written as a mixture of two kernels, i.e. 
\[\mathbb{L}^{u,y}(a|\cdot)= \zeta \mathbb{L}_1^{u,y}(a|\cdot) + (1-\zeta)\mathbb{L}_2^{u,y}(a|\cdot),\]
each of which is induced by a different Metropolis-Hastings algorithm, and $\zeta$ determines the mixture weight. The proposal mechanism for each of the kernels corresponds to a different type of movement, described next:

\begin{enumerate}
\item For $\mathbb{L}^{u,y}_1(a|\cdot) $ we use Metropolis-Hastings to sample from the distribution $\nu_{a|u,y}$ using the following proposal: given $a=(k,\theta)$ with $\theta= [\theta_1, \dots, \theta_k]$ we set $\tilde{k}=k$ (i.e. the number of grid points stays the same) and let $\tilde{\theta}$ be defined by
\[ \tilde{\theta}_i= \theta_i , \quad i=1, \dots, k-1,\]
and sample $\tilde{\theta}_k$ from a distribution on $D$ with density (w.r.t. Lebesgue measure on $D$) $\tau_\theta$. In principle the density used to sample $\tilde{\theta}_k$ may depend on $\theta$.

\item For $\mathbb{L}^{u,y}_2(a|\cdot)$ we use Metropolis-Hastings to sample from the distribution $\nu_{a|u,y}$ using the following proposal: given $a=(k,\theta)$ we sample $\tilde{k} \sim \sigma(k|\cdot)$ where $\sigma(k|\cdot)$ is a Markov kernel on $\N$, and then generate $\tilde{\theta}$ according to
\begin{itemize}
\item If $\tilde{k} > k$ let $\tilde{\theta}_i=\theta_i$ for all $i=1, \dots,k$ and then sample $\tilde{\theta}_{k+1} , \dots, \tilde{\theta}_{\tilde{k} }$ independently from the density $\tau_\theta$.  
\item If $\tilde{k}  \leq k$ let $\tilde{\theta}_i=\theta_i$ for all $i=1, \dots, \tilde{k} $.
\end{itemize}
\end{enumerate}

\begin{remark}
We notice that the proposals described above are particular cases of the ones used in reversible jump Markov chain Monte Carlo \cite{green1995reversible}.
\end{remark}


For the Metropolis-Hastings algorithm associated to $\mathbb{L}^{u,y}_1(a|\cdot)$ the acceptance probability takes the form
\[ \alpha_1(a, \tilde{a})= \min \biggl\{ 1, \exp\Bigl(- \Psi(u,\tilde{a};y) + \Psi(u,a;y)\Bigr)
\frac{ \tau_{\tilde{\theta}}(\theta_k) }{ \tau_{\theta}( \tilde{\theta}_k)} \biggr\}, \]
where recall $a=(k,\theta)$ and $\theta=(\theta_1, \dots, \theta_k)$ and $\tilde{a}$ is defined similarly.

For the Metropolis-Hastings algorithm associated to $\mathbb{L}^{u,y}_2(a|\cdot)$ the acceptance probability takes the form
\[ \alpha_2(a, \tilde{a})= \min \biggl\{ 1, \frac{\sigma(k| \tilde{k} ) \nu_{k}(\tilde{k} ) }{\sigma(\tilde{k} |k ) \nu_{k}(k)}\exp\Bigl(- \Psi(u,\tilde{a};y) + \Psi(u,a;y)\Bigr) H(k,\theta,\tilde{k}, \tilde{\theta}) \biggr\}, \]
where
\[ H(k, \theta, \tilde{k}, \tilde{\theta} ):= \begin{cases}  
 \prod _{i=1 }^{k - \tilde{k} } \tau_{\tilde{\theta}}( \theta_{\tilde{k} +i} )
 & \text{ if } k >\tilde{k}, \\
 \left( \prod _{i=1 }^{ \tilde{k}  - k} \tau_{\theta}( \tilde{\theta}_{k+i} ) \right)^{-1} & \text{ if } \tilde{k} \geq k. \end{cases}\]

We notice that since each of the kernels $\mathbb{L}_1^{u,y}(a|\cdot)$ and $\mathbb{L}_2^{u,y}(a|\cdot)$ is defined by a Metropolis-Hastings algorithm, they leave  the target $\nu_{a|u,y}$ invariant, and hence so does the kernel $\mathbb{L}^{u,y}(a|\cdot)$.

 \begin{remark}
  If in the above the distribution $\tau_{\theta}$ is, regardless of $\theta$, the uniform distribution on the domain $D$, then the acceptance probabilities reduce, respectively, to
\[ \alpha_1(a,\tilde{a})= \min\Bigr\{1, \exp\Bigl(- \Psi(u, \tilde{a};y )+ \Psi(u, a;y) \Bigr)\Bigr\}, \]
 and
  \[ \alpha_2(a, \tilde{a} )= \min \left \{1, \frac{\sigma(k| \tilde{k}) \nu_{k}(\tilde{k}) }{\sigma(\tilde{k}|k ) \nu_{k}(k)} \exp\Bigl(- \Psi(u,\tilde{a};y) + \Psi(u,a;y) \Bigr) \right \}. \]
 \end{remark}

	\subsubsection{Sampling Density-based Discretizations}
Since in this case the dimension of $\theta$ is fixed, the calculation of the acceptance probabilities is straightforward and the details are omitted. We refer to Subsection \ref{ssec:FEM} for a numerical example.

\section{Numerical Examples}\label{sec:numerics}
In this section we demonstrate the applicability of our framework and sampling approach in a variety of inverse problems. Our aim is illustrating the benefits and potential limitations of the methods; for this reason we consider inverse problems for which we have intuitive understanding of where the discretizations should concentrate, thus validating the performance of the proposed approach. Before discussing the numerical results, we summarize the main goals and outcomes of each set of experiments:
\begin{itemize}
\item In Subsection \ref{section:euler ode 1} we consider  an inverse problem in mechanics  \cite{bigoni2019greedy}, for which some observation settings highly influence the best choice of discretization while others inform it mildly. Our numerical results show that the gain afforded by grid learning is most clear whenever the observation locations highly influence the choice of discretization. 
We employ grid-based discretizations as described in Example \ref{ex:gridbased} with an Euler discretization of the forward map. We also illustrate the applicability of the method in both  finite and infinite-dimensional representations  of the unknown parameter, showing a more dramatic effect in the latter.
\item In Subsection \ref{sec:SDE} we consider an inverse problem in signal processing \cite{hairer2011signal}, with a choice of observation locations that determine where the discretization should concentrate. Our numerical results show that the grids adapt to the expected region, and that the degrees of freedom in the discretization necessary to reconstruct the unknown is below that necessary to satisfy stability of the numerical method with uniform grids.
We employ grid-based discretizations as described in Example \ref{ex:gridbased} with an Euler-Maruyama discretization of the forward map. 
\item In Subsection \ref{ssec:FEM} we consider an inverse problem in source detection, where the true hidden unknown determines how best to discretize the forward model. Our numerical results show that the grids adapt as expected.
We employ density-based discretizations as described in Example \ref{ex:densitybased} with a finite element discretization of the forward model.
\end{itemize}
	
	\subsection{Euler Discretization of ODEs: Estimation of the Young’s Modulus of a Cantilever Beam}\label{section:euler ode 1}
	We consider an inhomogeneous cantilever beam clamped on one side ($x=0$) and free on the other ($x=L$). Define $D=[0, L].$ Let  $u(x)$ denote its Young's modulus and let $M(x)$ be a load applied onto the beam. 
	Timoshenko's beam theory gives the displacement $z(x)$ of the beam and the angle of rotation $\varphi(x)$ through the coupled ordinary differential equations
	\begin{equation}
	\label{eq:cantilever}
	\begin{cases}
	\frac{d}{dx} \Bigl[\frac{u(x)}{2(1+r)} \Bigl(\varphi(x)-\frac{d}{dx}z(x) \Bigr) \Bigr]=\frac{M(x)}{\kappa A}, \\
	\frac{d}{dx} \Bigl(u(x)I\frac{d}{dx}\varphi(x)\Bigr)=\kappa A \frac{u(x)}{2(1+r)} \Bigr(\varphi(x)-\frac{d}{dx}z(x)\Bigr),
	\end{cases}
	\end{equation}
	where $r$, $\kappa$, $A$, $I$ are physical constants. Following \cite{bigoni2019greedy}, we consider the inverse problem  of estimating the Young's modulus $u(x)$ from sparse observations of the displacement $z(x)$, where both $u$ and $z$ are functions from $D$ to $\mathbb{R}.$
	
	Let $\F:u\mapsto z$ be the solution map to equations \eqref{eq:cantilever}. Let $\{s_i\}_{i=1}^m \subset D$ be the locations of the observation sensors, leading to the  observation operator $\O: z \mapsto y \in \mathbb{R}^m $ defined coordinate-wise by 
	$$O_i(z):=\int_0^{L}z\phi_i\mbox{d}x, \quad \quad \phi_i(x):= \frac{1}{\gamma_i} \exp\Bigl(-(s_i-x)^2/(2\delta^2)\Bigr), \quad \quad 1 \le i \le m ,$$
	 where  $\delta=10^{-4}$ and $\gamma_i$ is the normalizing constant such that $\int_0^L \phi_i\mbox{d}x=1$. Data are generated according to the model 
	 $$y= \mathcal{O}\circ\F(u)+\eta=:\mathcal{G}(u)+\eta,$$
	 where $\eta$ denotes the observation error, which is assumed to follow a Gaussian distribution $N(0, \gamma_{obs}^2I)$. Notice that for system \eqref{eq:cantilever} with proper boundary conditions specified at $x=0$, the displacement  $z(x^\star)$ at any point $0 <x^\star<L$  depends only on the values $\{u(x):x<x^\star\}.$ Thus, we expect suitable discretizations of the forward model to  refine finely only the region $\{0 < x  < s_m\},$ where $s_m$ is the right-most observation location. We will discuss this in detail in section \ref{sec:beamimplem}.

	\subsubsection{Forward Discretization}

	\begin{figure}
		\centering
		\begin{subfigure}{.33\textwidth}
			\centering
			\includegraphics[height = 3.5cm,width=1\linewidth]{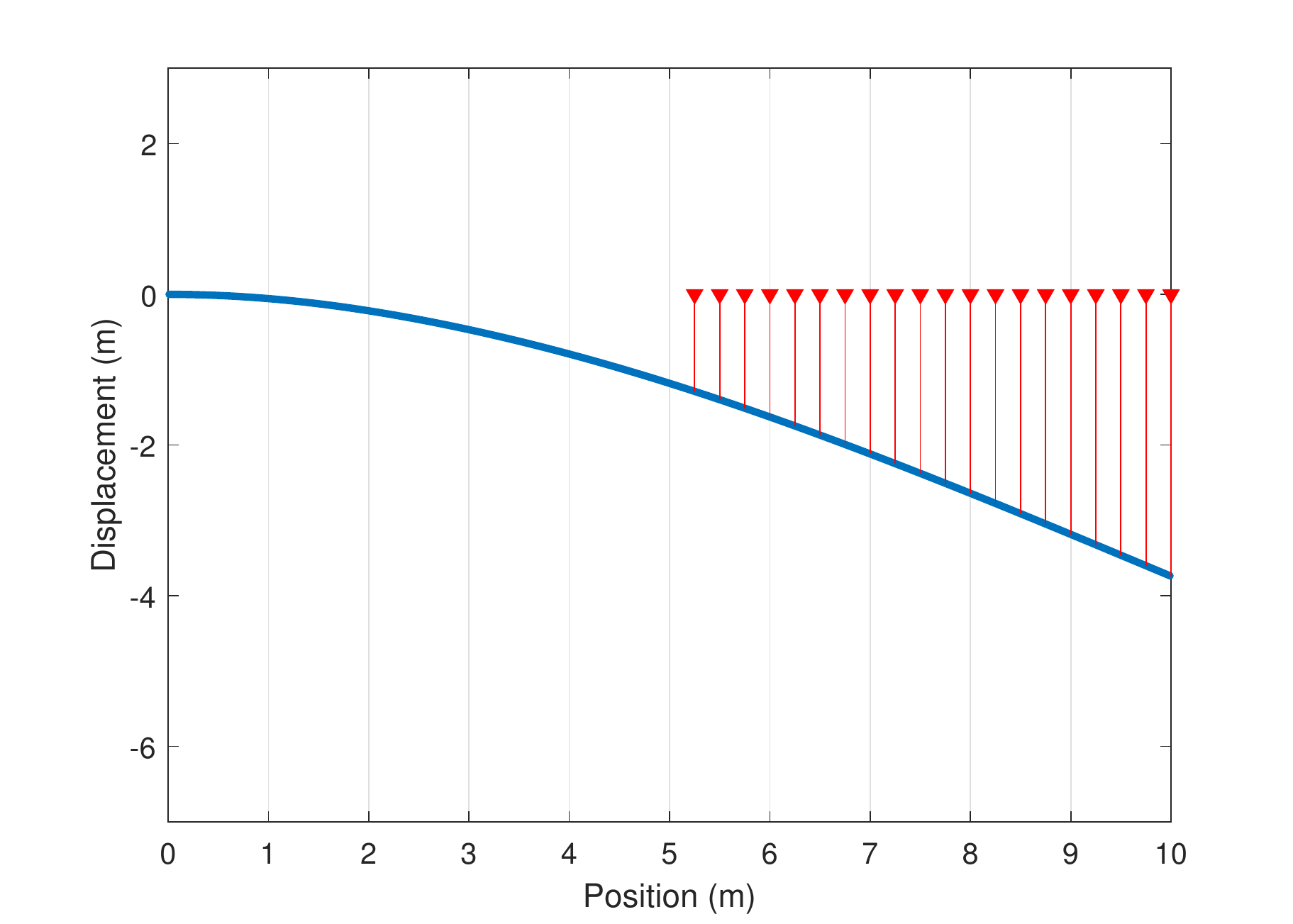}
			\caption[short]{Observation locations.}
			\label{fig:rdis1}
		\end{subfigure}%
		\hfill
		\begin{subfigure}{.33\textwidth}
			\centering
			\includegraphics[height = 3.5cm,width=1\linewidth]{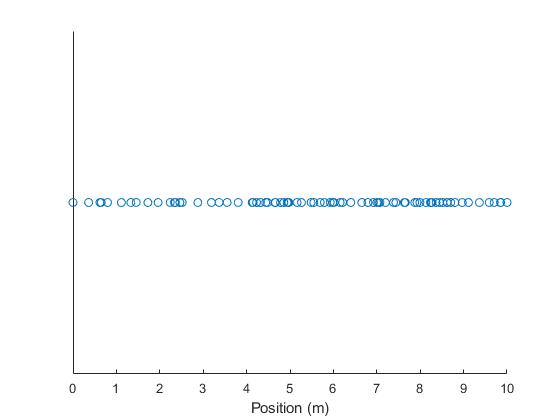}
			\caption[short]{A sampled grid.}
			\label{fig:rdis2}
		\end{subfigure}
		\hfill
		\begin{subfigure}{.33\textwidth}
			\centering
			\includegraphics[height = 3.5cm,width=1\linewidth]{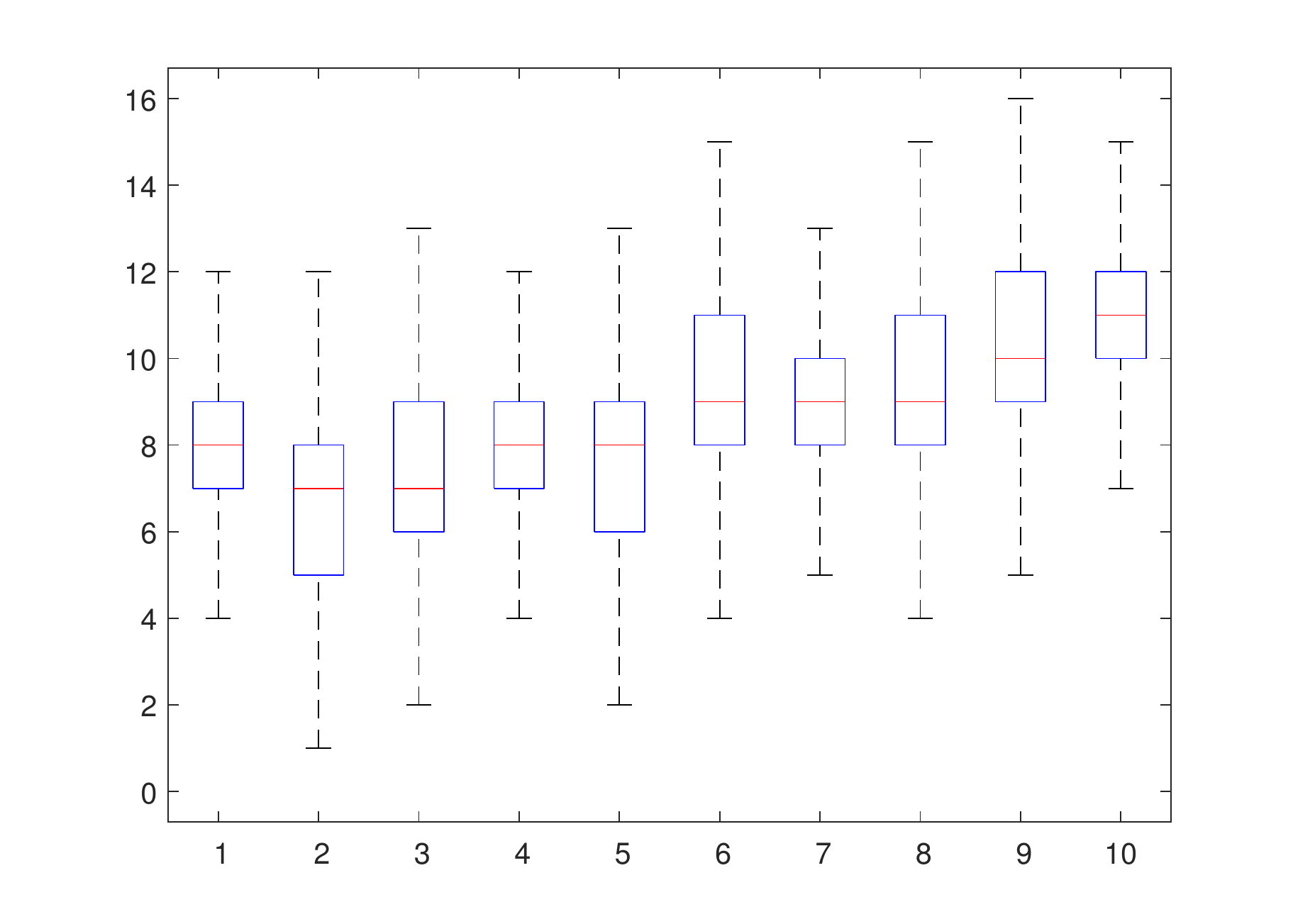}
			\caption[short]{Number of grid points.}
			\label{fig:rdis3}
		\end{subfigure}
		\vskip\baselineskip
		\begin{subfigure}{.45\textwidth}
			\centering
			\includegraphics[height = 4.5cm,width=8.25cm]{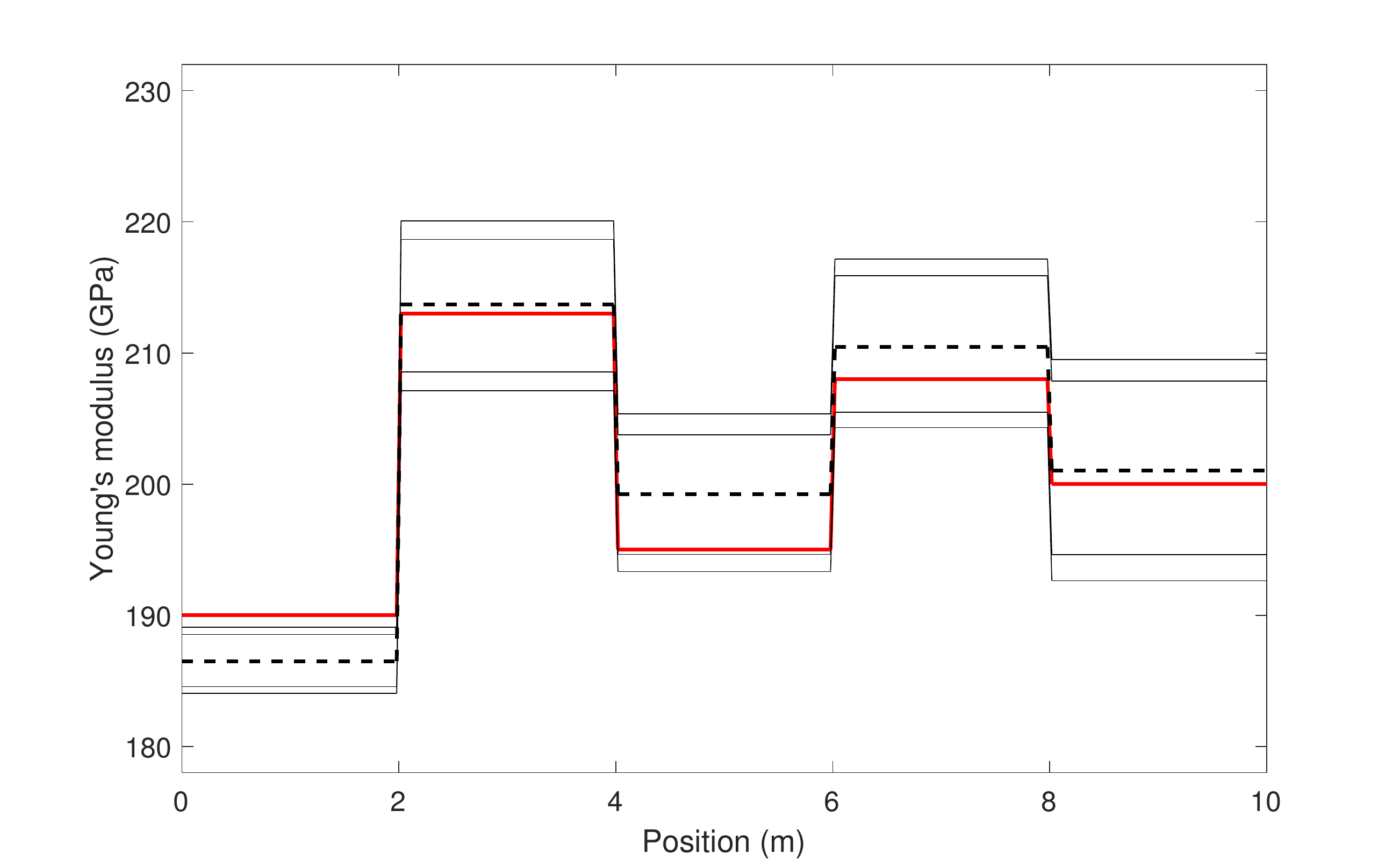}
			\caption{True vs. posterior with grid learning.}
			\label{fig:rdis4}
		\end{subfigure}
		\hfill
		\begin{subfigure}{.45\textwidth}
			\centering
			\includegraphics[height = 4.5cm,width=8.25cm]{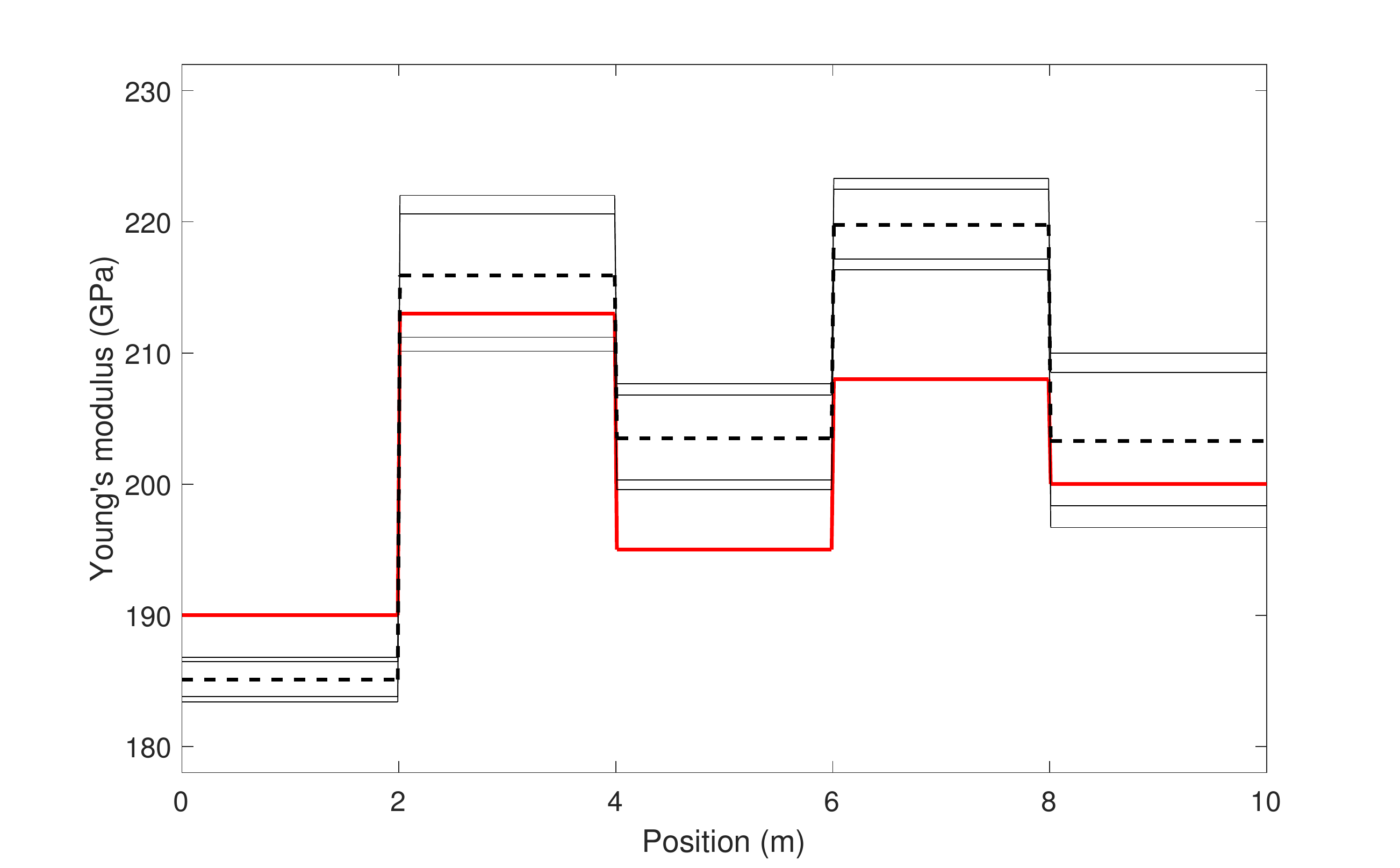}
			\caption{True vs. posterior without grid learning.}
			\label{fig:rdis5}
		\end{subfigure}
	
		\centering
		\begin{subfigure}{.33\textwidth}
			\centering
			\includegraphics[height = 3.5cm,width=1\linewidth]{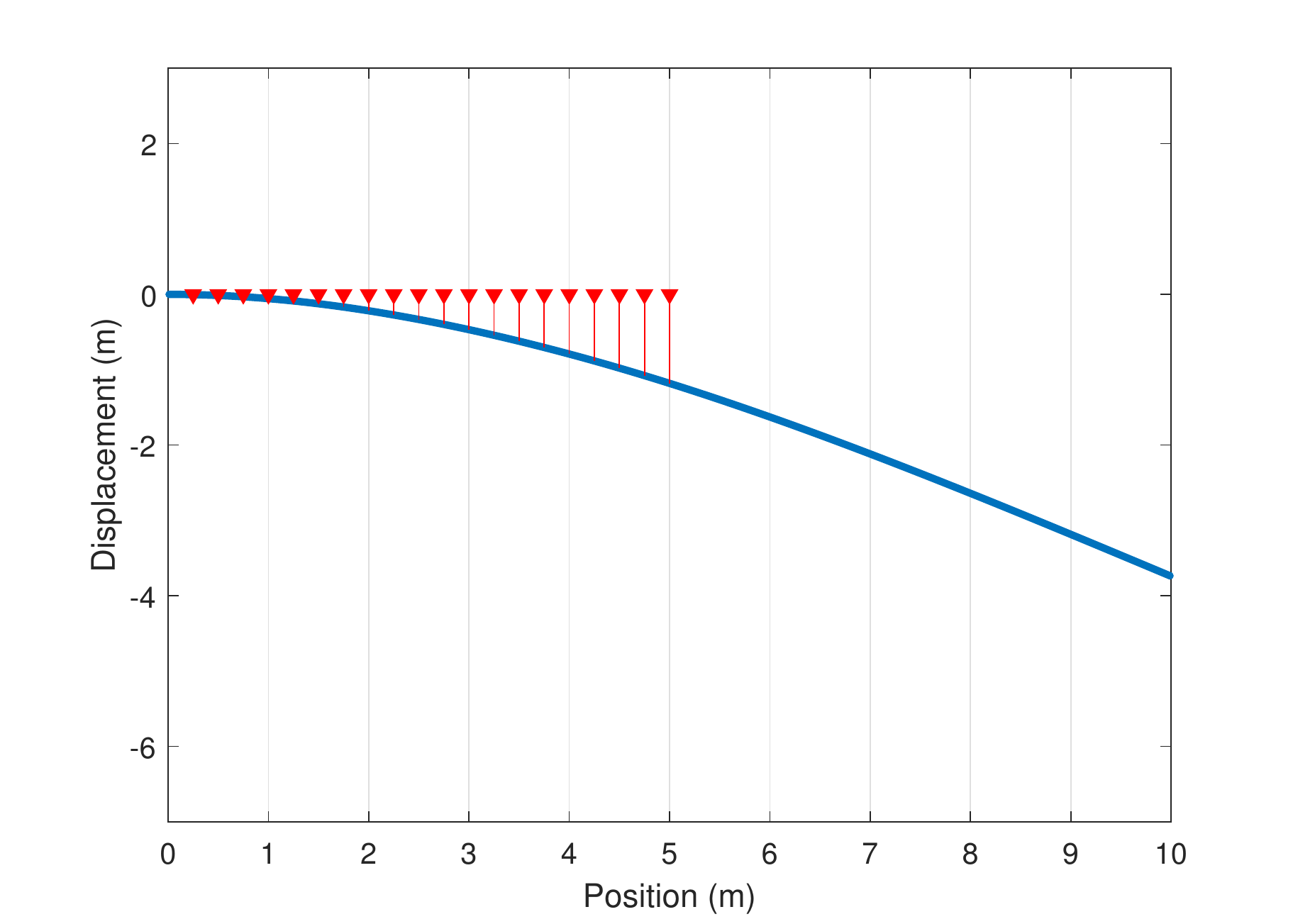}
			\caption{Observation locations.}
			\label{fig:ldis1}
		\end{subfigure}%
		\hfill
		\begin{subfigure}{.33\textwidth}
			\centering
			\includegraphics[height = 3.5cm,width=1\linewidth]{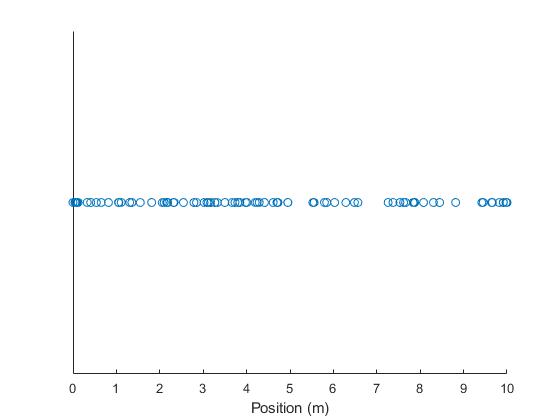}
			\caption{A sampled grid.}
			\label{fig:ldis2}
		\end{subfigure}
		\hfill
		\begin{subfigure}{.33\textwidth}
			\centering
			\includegraphics[height = 3.5cm,width=1\linewidth]{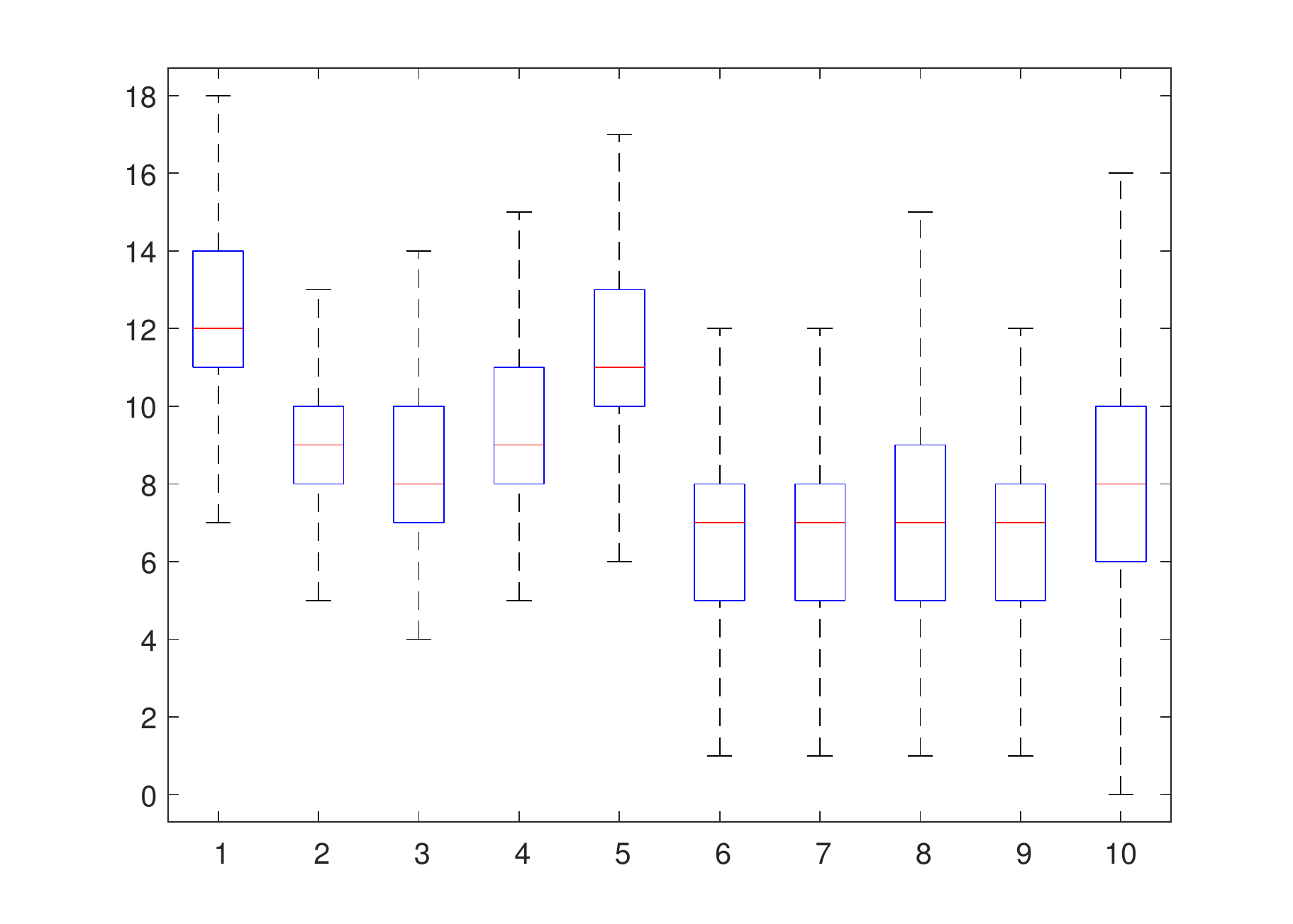}
			\caption{Number of grid points.}
			\label{fig:ldis3}
		\end{subfigure}
		\vskip\baselineskip	
		\begin{subfigure}{.45\textwidth}
			\centering
			\includegraphics[height = 4.5cm,width=8.25cm]{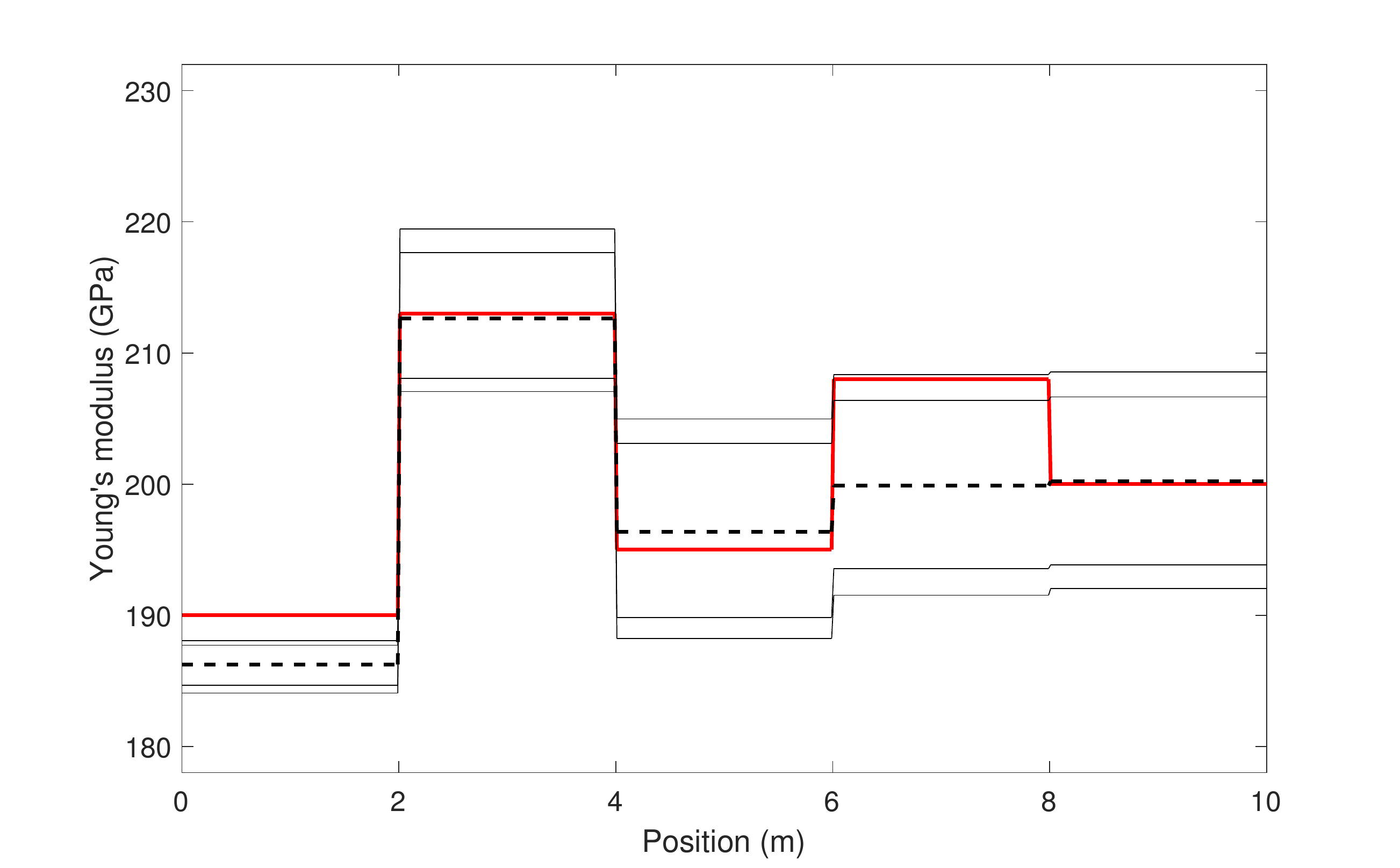}
			\caption{True vs. posterior with grid learning.}
			\label{fig:ldis4}
		\end{subfigure}
		\hfill
		\begin{subfigure}{.45\textwidth}
			\centering
			\includegraphics[height = 4.5cm,width=8.25cm]{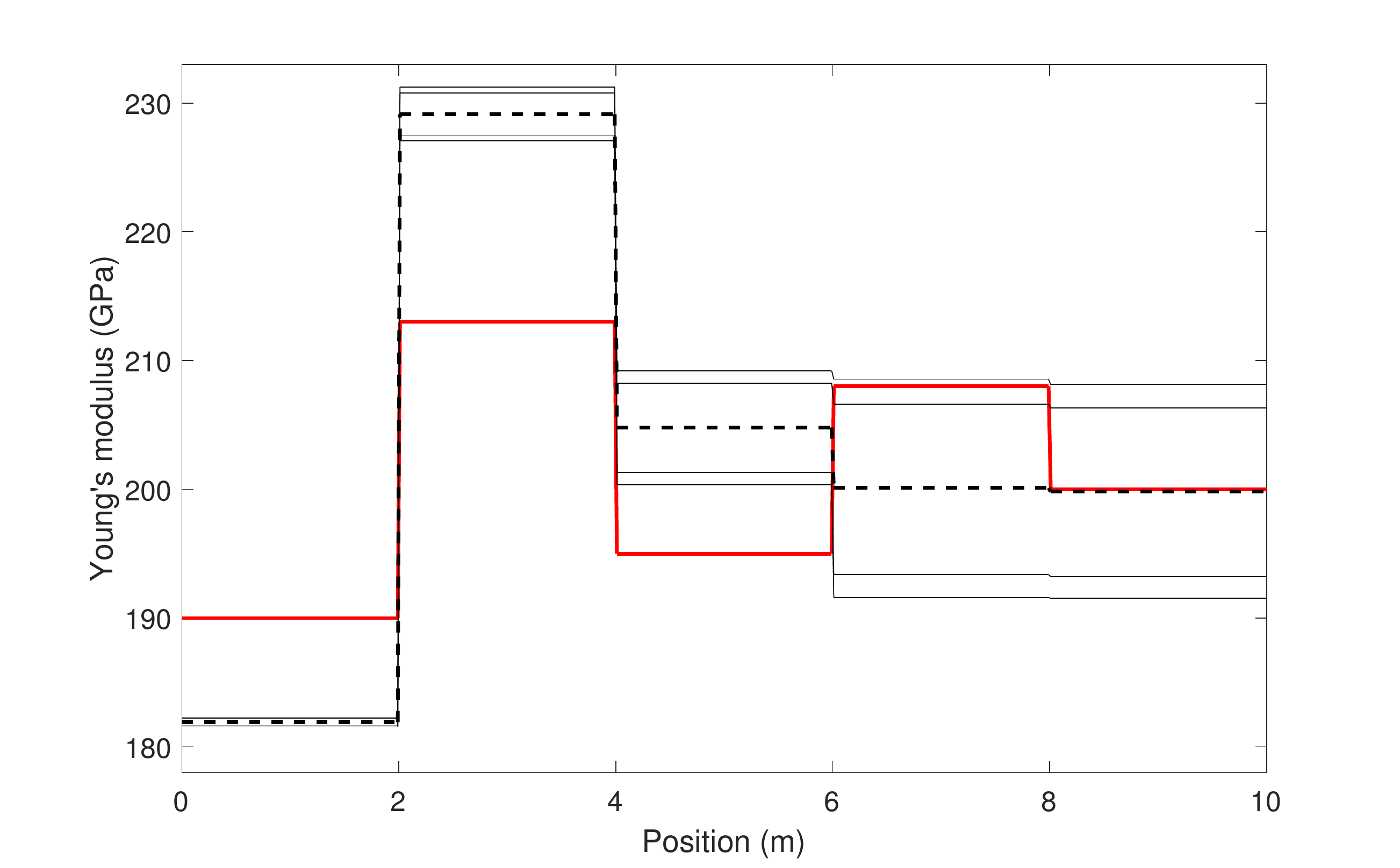}
			\caption{True vs. posterior without grid learning.}
			\label{fig:ldis5}
		\end{subfigure}
\caption{Reconstruction of a piece-wise constant Young's modulus. Two settings for the observation locations are considered, shown in Figures \ref{fig:rdis1}, \ref{fig:ldis1}. For each setting, Figures \ref{fig:rdis2} and \ref{fig:ldis2} show one sample from the marginal distribution $q_{a|y}(a)$ simulated by MCMC. Figures \ref{fig:rdis3} and \ref{fig:ldis3} report box-plots with the number of grid points that fall in each subinterval $[i-1,i]$, $i=1,\dots,10$. Figures \ref{fig:rdis4}, \ref{fig:ldis4} show the mean (dashed black) and the 5, 10, 90, 95-percentiles (thin black) of the marginal $q_{u|y}(u)$, versus the true value (red), with data-driven forward discretization. Figures \ref{fig:rdis5}, \ref{fig:ldis5} show the same results with a fixed uniform-grid discretization.}
		\label{fig:cantileverdisc}
	\end{figure}

			\begin{figure}
		\centering
        \begin{subfigure}{.33\textwidth}
			\centering
			\includegraphics[height = 3.5cm,width=1\linewidth]{beam/result1-b-eps-converted-to.pdf}
			\caption{Observation locations.}
			\label{fig:rcont1}
		\end{subfigure}%
		\hfill
		\begin{subfigure}{.33\textwidth}
			\centering
			\includegraphics[height = 3.5cm,width=1\linewidth]{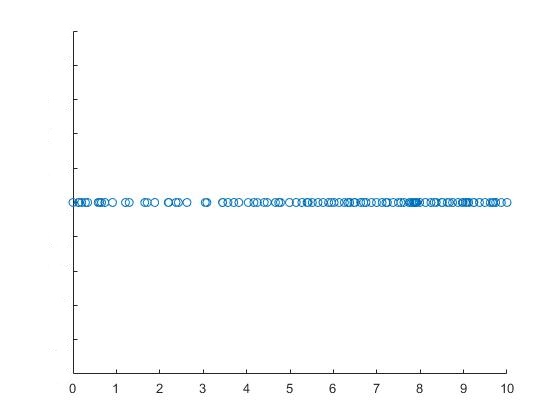}	
			\caption{A sampled grid.}
			\label{fig:rcont2}
		\end{subfigure}
		\hfill
		\begin{subfigure}{.33\textwidth}
			\centering
			\includegraphics[height = 3.5cm,width=1\linewidth]{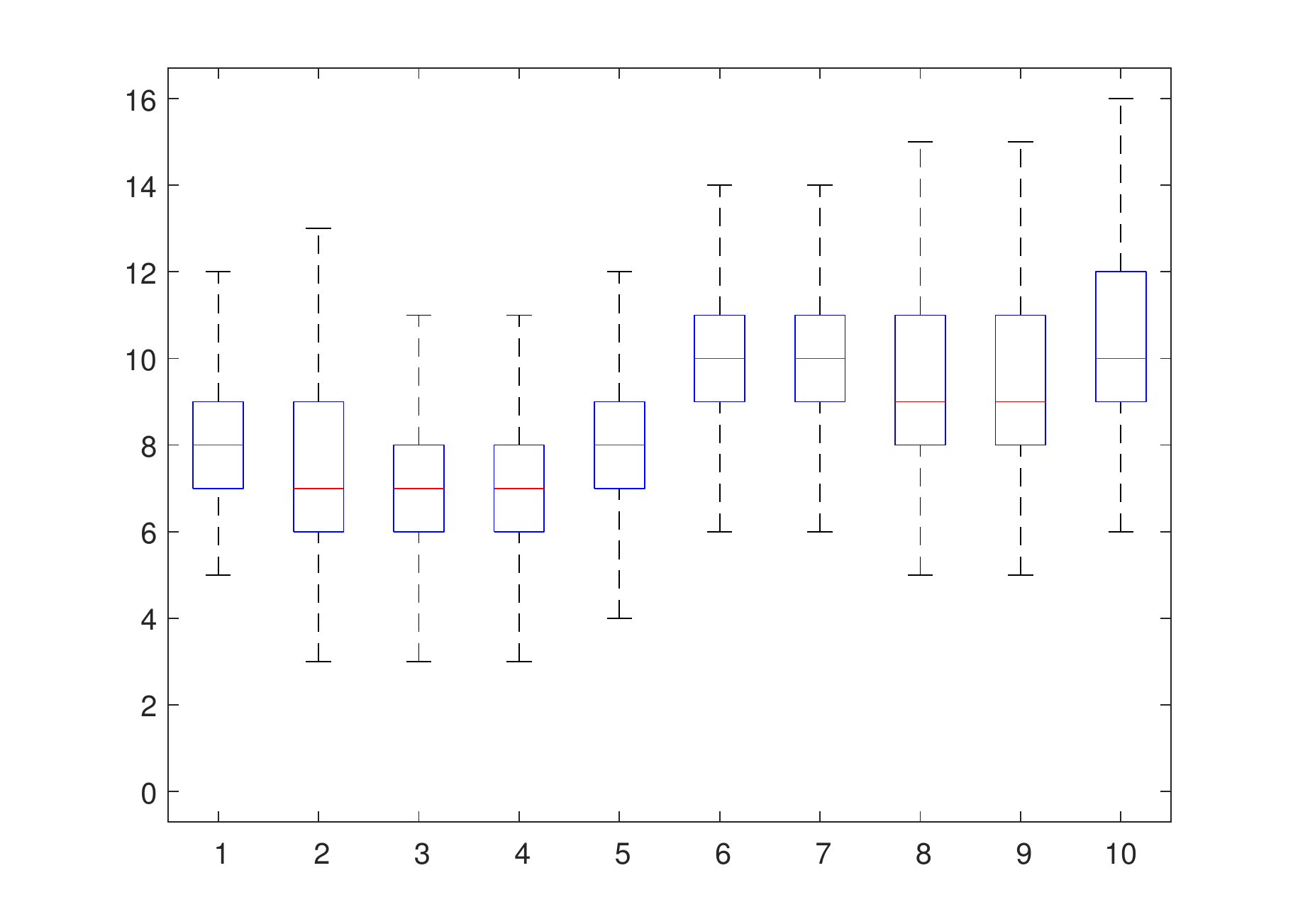}
			\caption{Number of grid points.}
			\label{fig:rcont3}
		\end{subfigure}
		\vskip\baselineskip

		\begin{subfigure}{.45\textwidth}
			\centering
			\includegraphics[height = 4.5cm,width=8.25cm]{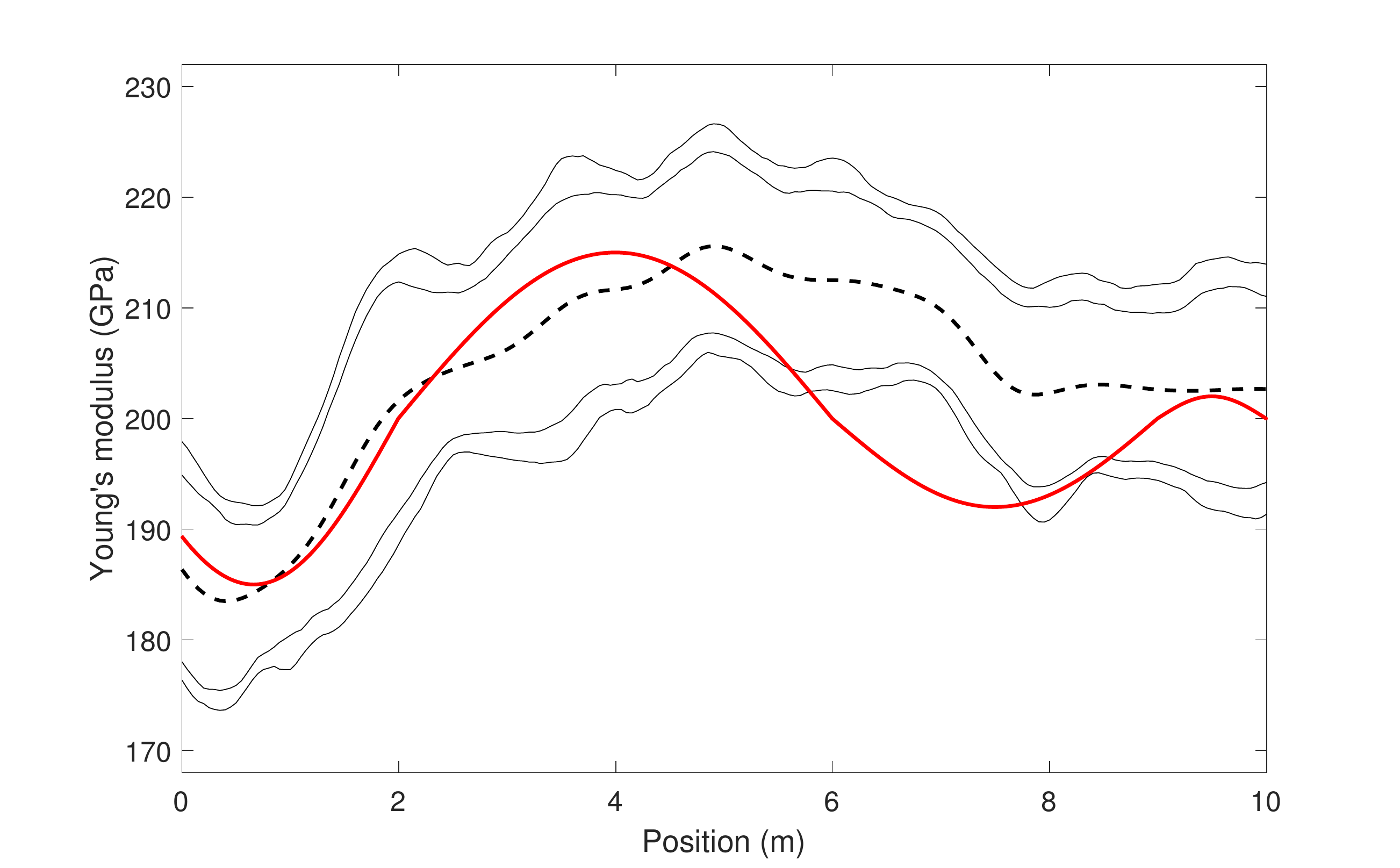}
			\caption{True vs. posterior with grid learning.}
			\label{fig:rcont4}
		\end{subfigure}
		\hfill
		\begin{subfigure}{.45\textwidth}
			\centering
			\hspace{-1cm}
			\includegraphics[height = 4.5cm,width=8.25cm]{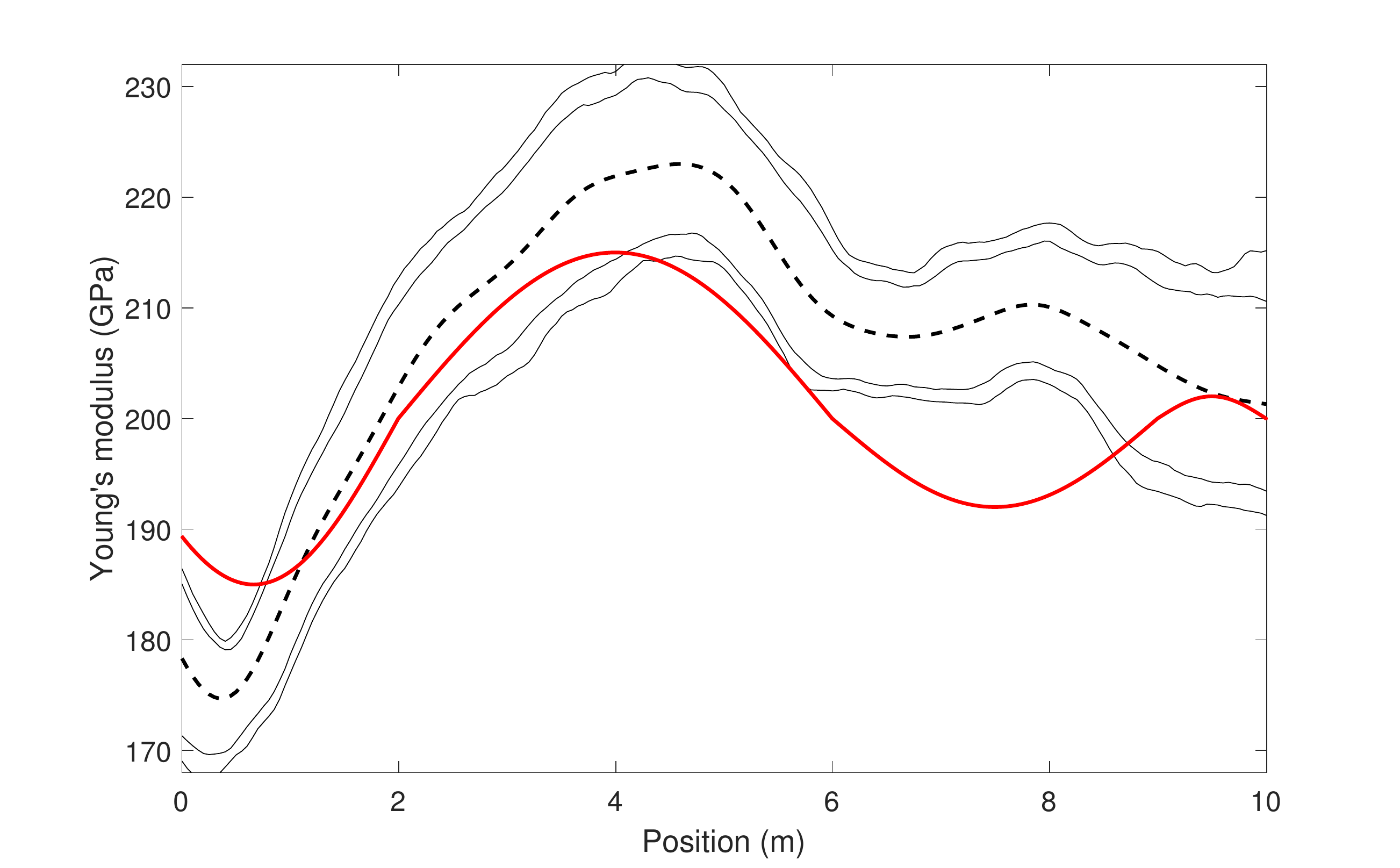}	
			\caption{True vs. posterior without grid learning.}
			\label{fig:rcont5}
		\end{subfigure}
			\centering
        \begin{subfigure}{.33\textwidth}
		\centering
		\includegraphics[height = 3.5cm,width=1\linewidth]{beam/result2-b-eps-converted-to.pdf}
		\caption{Observation locations.}
		\label{fig:lcont1}
	\end{subfigure}%
	\hfill
	\begin{subfigure}{.33\textwidth}
		\centering
		\includegraphics[height = 3.5cm,width=1\linewidth]{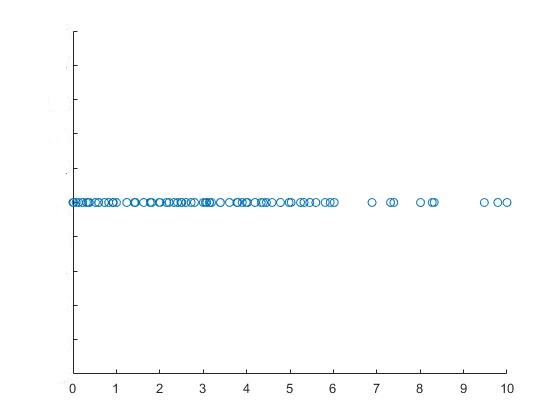}
		\caption{A sampled grid.}
		\label{fig:lcont2}
	\end{subfigure}
	\hfill
	\begin{subfigure}{.33\textwidth}
		\centering
		\includegraphics[height = 3.5cm,width=1\linewidth]{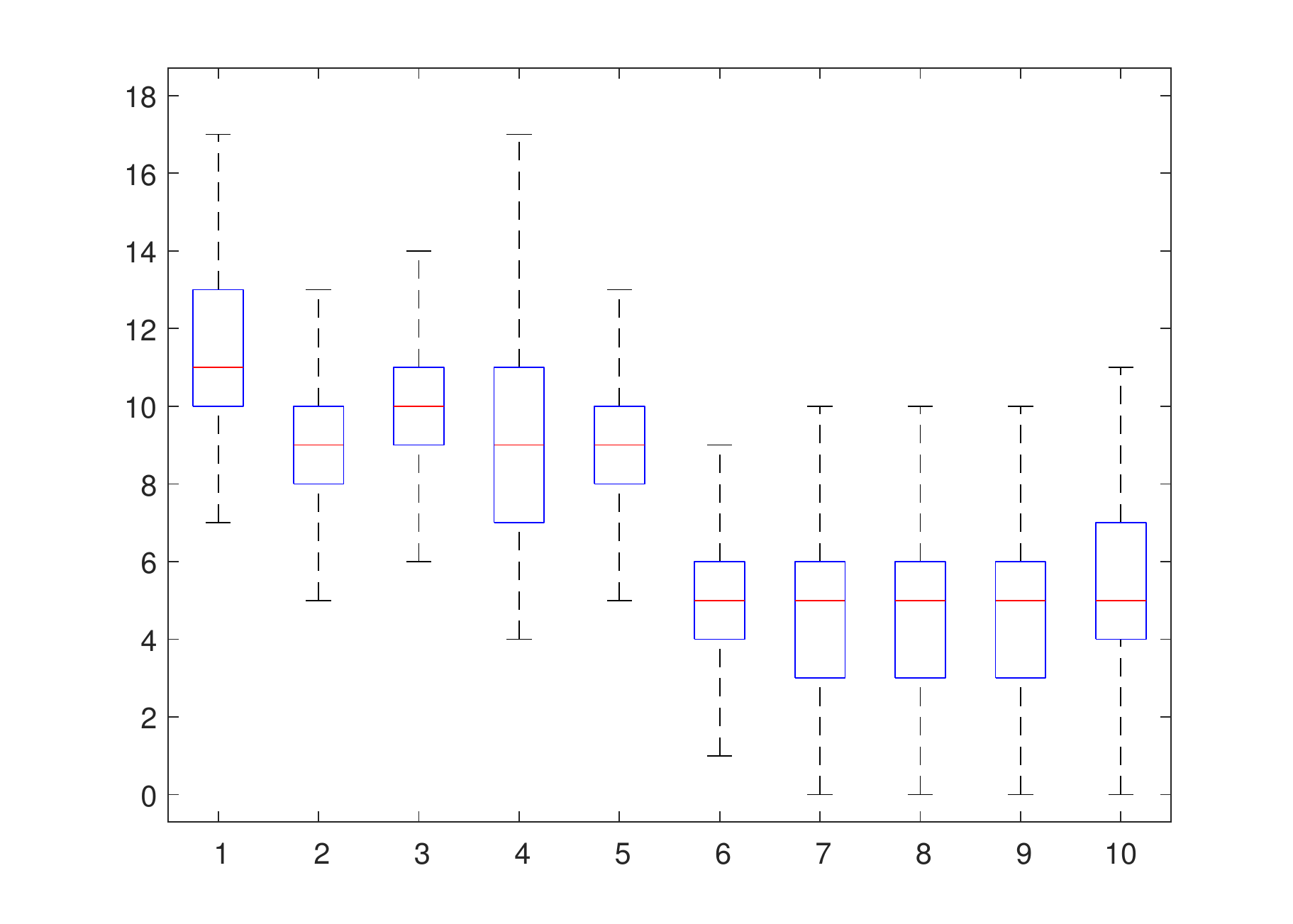}
		\caption{Number of grid points.}
		\label{fig:lcont3}
	\end{subfigure}
	\vskip\baselineskip

	\begin{subfigure}{.45\textwidth}
		\centering
		
		\includegraphics[height = 4.5cm,width=8.25cm]{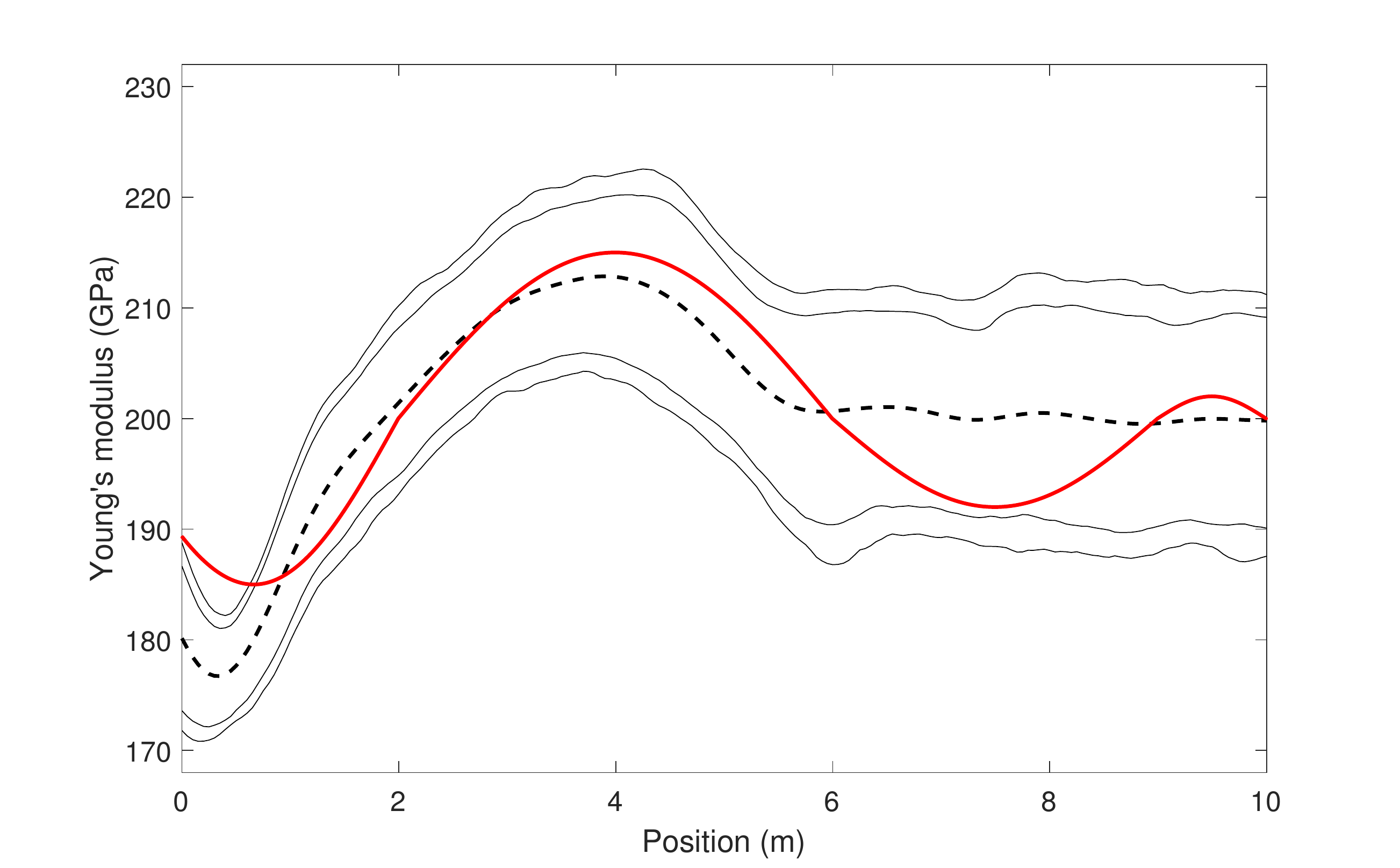}
		\caption{True vs. posterior with grid learning.}
		\label{fig:lcont4}
	\end{subfigure}
	\hfill
	\begin{subfigure}{.45\textwidth}
	\hspace{-1cm}
		\centering
		\includegraphics[height = 4.5cm,width=8.25cm]{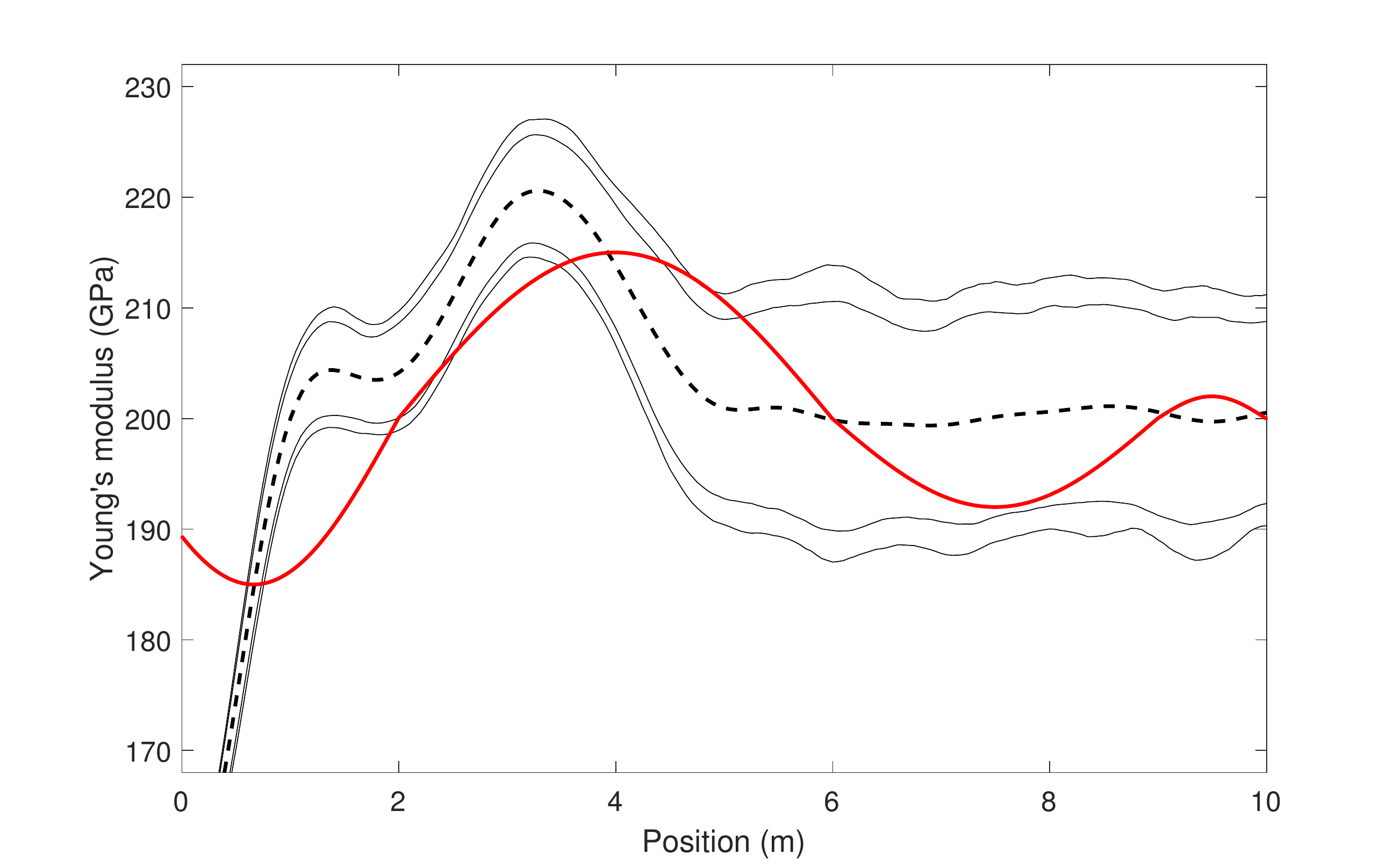}
		\caption{True vs. posterior without grid learning.}
		\label{fig:lcont5}
	\end{subfigure}
\caption{Reconstruction of a continuous Young's modulus. Two settings for the observation locations are considered, shown in Figures \ref{fig:rcont1}, \ref{fig:lcont1}. For each setting, Figures \ref{fig:rcont2} and \ref{fig:lcont2} show one sample from the marginal distribution $q_{a|y}(a)$ simulated by MCMC.  Figures \ref{fig:rcont3} and \ref{fig:lcont3} report box plots with number of grid points that fall in each subinterval $[i-1,i]$, $i=1,\dots,10$. Figures \ref{fig:rcont4}, \ref{fig:lcont4} show the mean (dashed black) and the 5, 10, 90, 95-percentiles (thin black) of the marginal $q_{u|y}(u)$, versus the true value (red), with data-driven forward discretization. Figures \ref{fig:rcont5}, \ref{fig:lcont5} show the same results with a fixed uniform-grid discretization.}
	\label{fig:cantilevercont}
	\end{figure}

	To solve system \eqref{eq:cantilever} we employ a finite difference method. A family of numerical solutions can be parameterized by the set
	\begin{equation*}
	\A:=\Bigl\{a=(k,\theta):k\in\mathcal{K}\subset\{1,2,\dots\},\theta=[x_1,\dots,x_k]\in[0,L]^k \Bigr\},
	\end{equation*}
	where $k$ is the number of grid points and $\theta$ are the grid locations. Precisely, for $a=(k,\theta)\in\A$, we first reorder $\theta$ so that 
	\begin{equation*}
	0=:x_0\leq x_1\leq \dots \leq x_k\leq x_{k+1}:=L
	\end{equation*}
	and we let $$\mathcal{F}^a:u\mapsto z^a$$ be the linearly interpolated explicit Euler finite difference solution to \eqref{eq:cantilever}, discretized using the ordered grid $\theta$. We also discretize the observation operator $\mathcal{O}$ using an Euler forward method, defined by
	\begin{equation*}
	 \mathcal{O}^a_i(z^a)=\sum_{j=0}^kz^a(x_j)\phi_i(x_j)(x_{j+1}-x_j).
	\end{equation*}
	Finally $\mathcal{G}$ is approximated by  $\mathcal{G}^a:=\mathcal{O}^a\circ\mathcal{F}^a.$	
	
	\subsubsection{Implementation Details and Numerical Results}\label{sec:beamimplem}
	For our numerical experiments we consider a beam of length $L=10\mbox{ m}$, width $w=0.1 \mbox{ m}$ and thickness $h=0.3\mbox{ m}$. We use a Poisson ratio $r=0.28$ and Timoshenko shear coefficient $\kappa=5/6$. $A=wh$ represents the cross-sectional area of the beam and $I=wh^3/12$ is the second moment of inertia. We run a virtual experiment of applying a point mass of $5\mbox{ kg}$ at the end of the beam, as seen in blue in Figures \ref{fig:rdis1} and \ref{fig:rcont1}. We assume that the observations are gathered with error $\gamma_{obs}^2=10^{-3}.$
	
	We first assume that the beam is made of 5 segments of different kinds of steel, each of length $2\mbox{ m}$, with corresponding Young's moduli $u^*=\{u_i^*\}_{i=1}^5=\{190, 213, 195, 208, 200\gpa\}$. The prior on $u \in \U = \R^5$ is given by $p_u(u) = \mathcal{N}(u;200 {\bf{1}}, 25 I_5)$ where $\bf{1}$ denotes the all-ones vector.  For this case we assume that the number of grid points $k$ is fixed to be $k=85$, i.e., the prior $p_k(k)$ is a point mass. The grid locations $\theta$ are assumed to be a priori uniformly distributed in $[0,L]^k$. Results are reported in Figure \ref{fig:cantileverdisc}. We next assume that the Young's modulus $u(x) \in \U = C([0,L]; \R)$ varies continuously with $x$. We set  a Gaussian process prior on $u$  defined by $\mu_u = \mathcal{GP}(200,c)$ with $c(x, x')=50\exp\Bigl(-(x-x')^2/0.5\Bigr)$. For this case we assume that the prior on the number of grid points $k$ follows a Poisson distribution with mean 60, i.e., $\nu_k(k)=$Poisson$(60)$, and the grid locations $\theta$ still have a uniform prior given $k$.\nc The true Young's modulus underlying the data and the reconstruction results are reported in Figure \ref{fig:cantilevercont}. 	Sampling is performed, both in the discrete and continuous settings, updating $u$ and $a$ alternately for a total number of $N=1.2\times 10^5$ iterations, with $\beta=0.08$, $\zeta=0.5$.

	\begin{figure}
	    \centering
	    \includegraphics[width=0.6\linewidth, height=4.75cm]{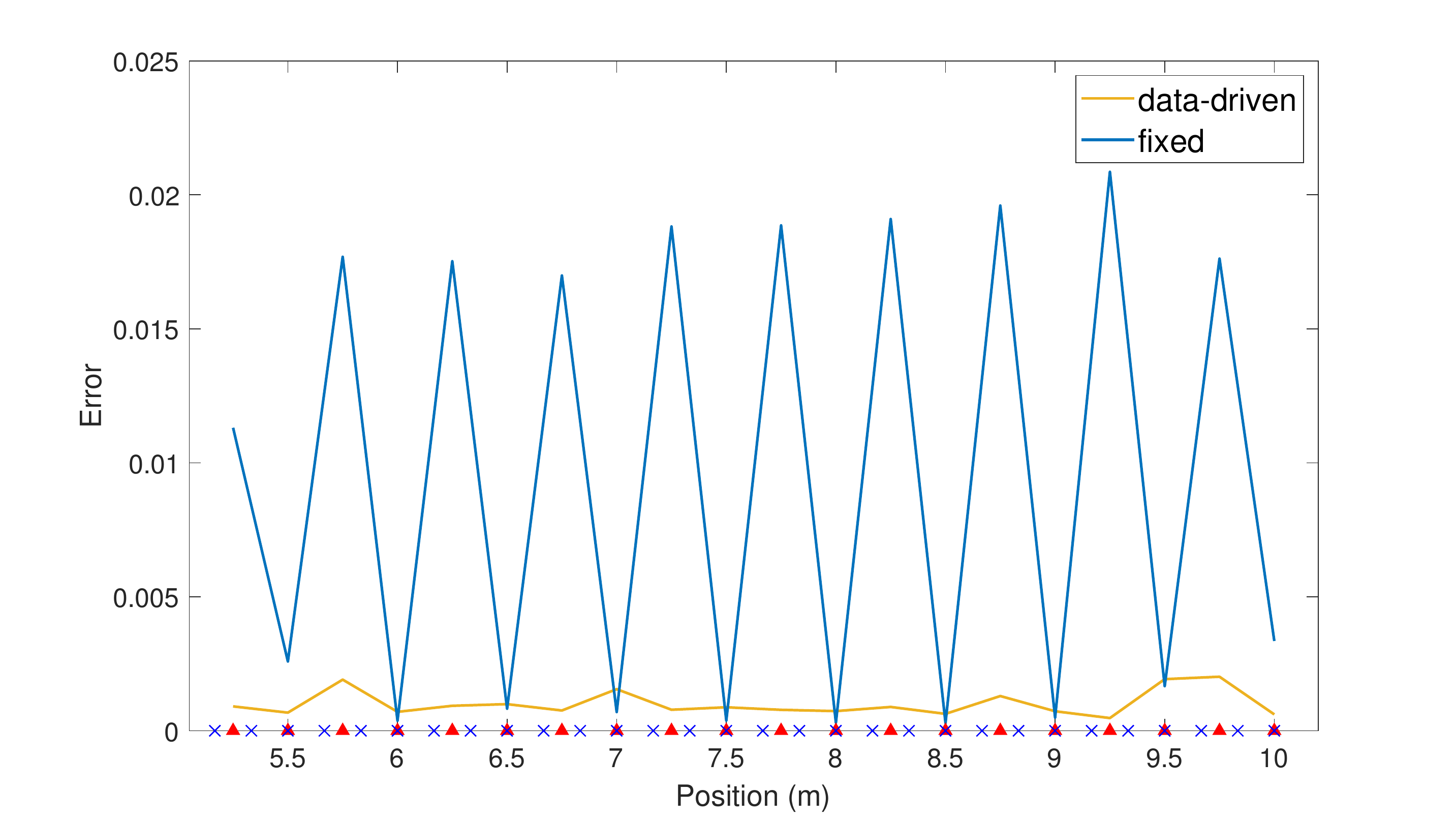}
	    \caption{The reconstruction error with fixed-grid discretization (blue) and with data-driven grid discretization (orange). Red triangles are observations locations, while blue crosses are the grid points used in the fixed-grid discretization.}
	    \label{fig:wiggly}
	\end{figure}

In both Figures \ref{fig:cantileverdisc} and \ref{fig:cantilevercont} two settings are considered. 
 In the first one observations are concentrated on the right side of the beam and in the second on the left.
 For reference, Figure \ref{figure:idealized} shows idealized  posteriors  considered, obtained with \emph{very} fine discretizations $k=500$, for each of the settings.  Notice that for system \eqref{eq:cantilever} with proper boundary conditions specified at $x=0$, the displacement  $z(x_0)$ at any point $0 <x_0<L$  depends only on the values $u(x)$ of Young's moduli with $x<x_0.$ This implies that when observations are gathered on the left side of the beam, the posterior on $u(x)$ agrees with the prior on the right-side, and no resources should be on discretizing the forward map on that region. In that case our adaptive data-driven discretizations are strongly concentrated on the left, as shown in Figures \ref{fig:ldis2}, \ref{fig:ldis3} and \ref{fig:lcont2}, \ref{fig:lcont3}. However, when observations are gathered on the right side of the beam, the data is informative on $u(x)$ for all $0<x<L.$ In such case, Figures \ref{fig:rdis2}, \ref{fig:rdis3} and \ref{fig:rcont2}, \ref{fig:rcont3} show that the data-driven discretizations are concentrated on the right, but less heavily so. See Tables \ref{tb:1}, \ref{tb:2} in the appendix for a more detailed description of the grid points distribution in both cases. Also, our results indicate that using data-driven discretizations will lead to a better estimation of the true Young's modulus, compared to fixed-grid discretizations. Additional results in the continuous Young's modulus setting are provided in the appendix. See Table \ref{accepttable} and Figure \ref{runningsample} for the averaged acceptence probability for $u$ and $a$, and history of MCMC samples of the high-dimensional $u$ at some fixed locations, indicating the stationarity of the Markov chain.

	Let $(u^{(n)},a^{(n)})$ be the output of the Gibbs sampling algorithm at iteration $n$. The \emph{reconstruction error} is defined as follows:
	\begin{equation}
		e_r=\sqrt{\sum_{n=1}^{N} \bigr|\G^{a^{(n)}}(u^{(n)})-\G(u)\bigl|^2\,\,}\,\,,
	\end{equation}
	where $\G(u)$ is approximately calculated on a very fine grid. In Figure \ref{fig:wiggly} we plot the reconstruction error for the second experiment where the Young's moduli is continuous and observations are gathered on the right-side. With fixed-grid discretization, the reconstruction error is small where the discretization matches the observation points. With adaptive data-driven discretizations the grid points will adaptively match the observation points in order to produce less error.

			\begin{figure}
		\centering
		\begin{subfigure}{.45\textwidth}
			\centering
			\includegraphics[height = 4.2cm,width=8cm]{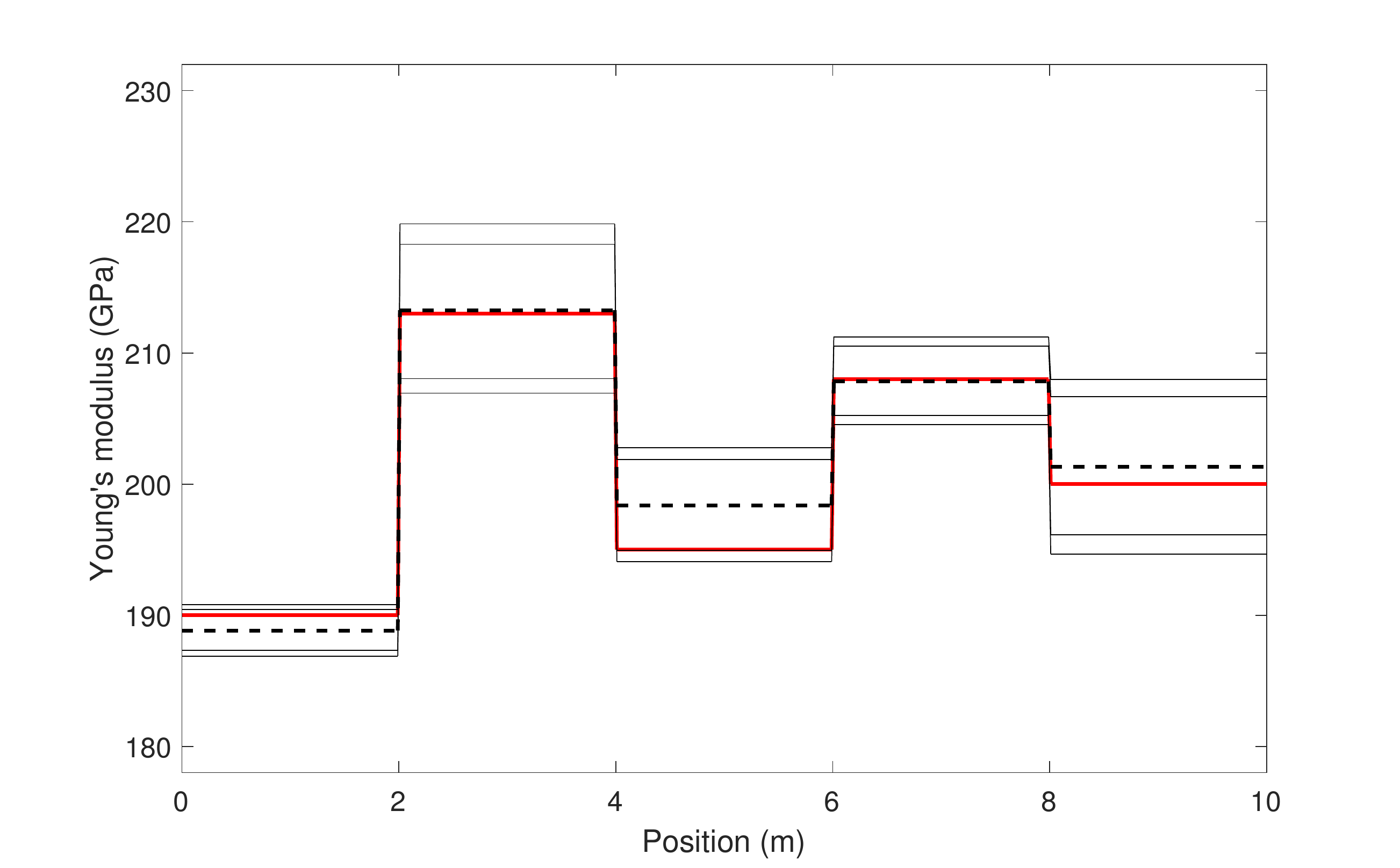}
			\caption{Piece-wise constant,  right observations.}
			\label{fig:i1}
		\end{subfigure}%
		\hfill
		\begin{subfigure}{.45\textwidth}
			\centering
			\includegraphics[height = 4.2cm,width=8cm]{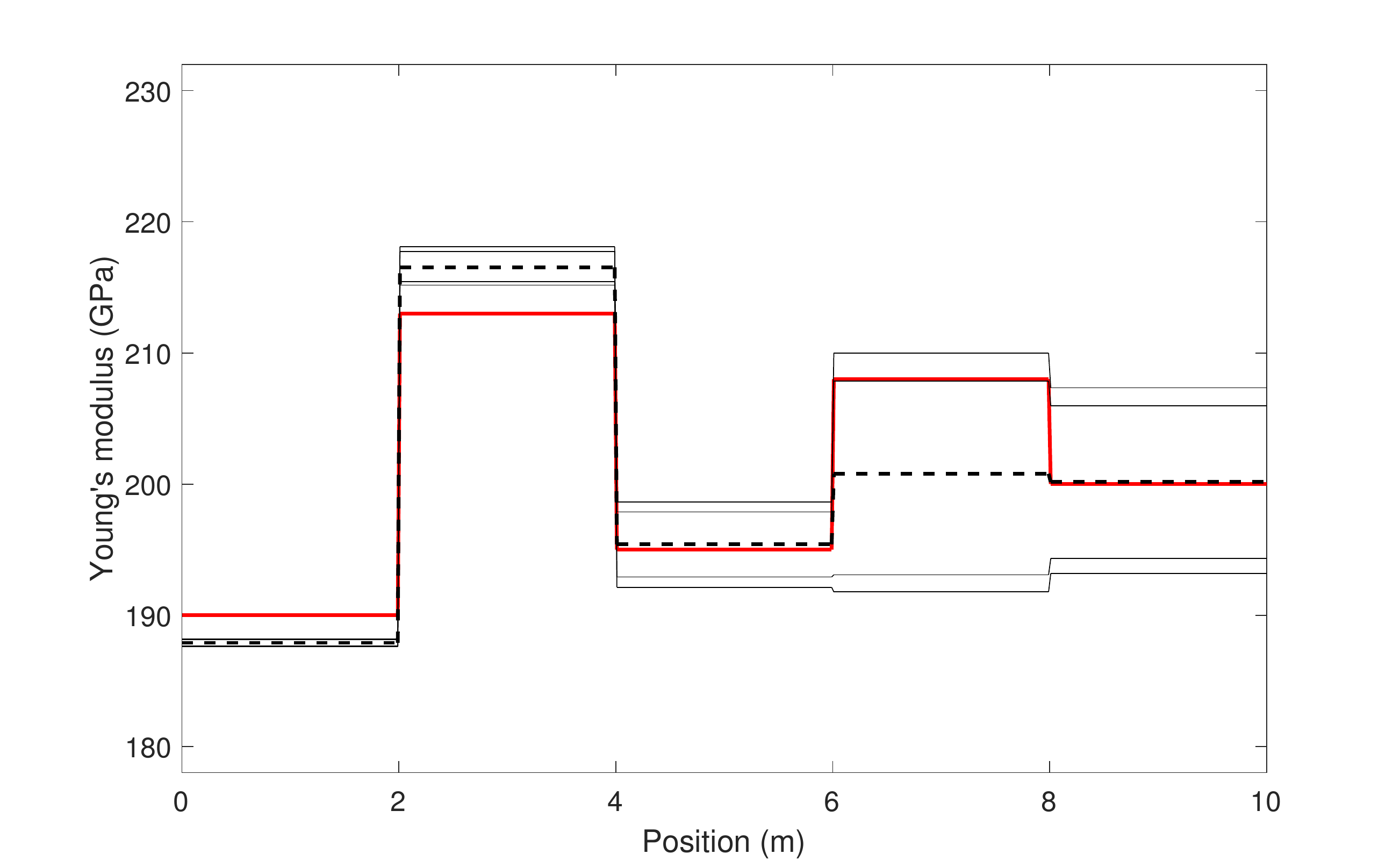}
			\caption{Piece-wise constant,  left observations.}
			\label{fig:i2}
		\end{subfigure}%
		\vskip\baselineskip
		\begin{subfigure}{.45\textwidth}
			\centering
			\includegraphics[height = 4.2cm,width=8cm]{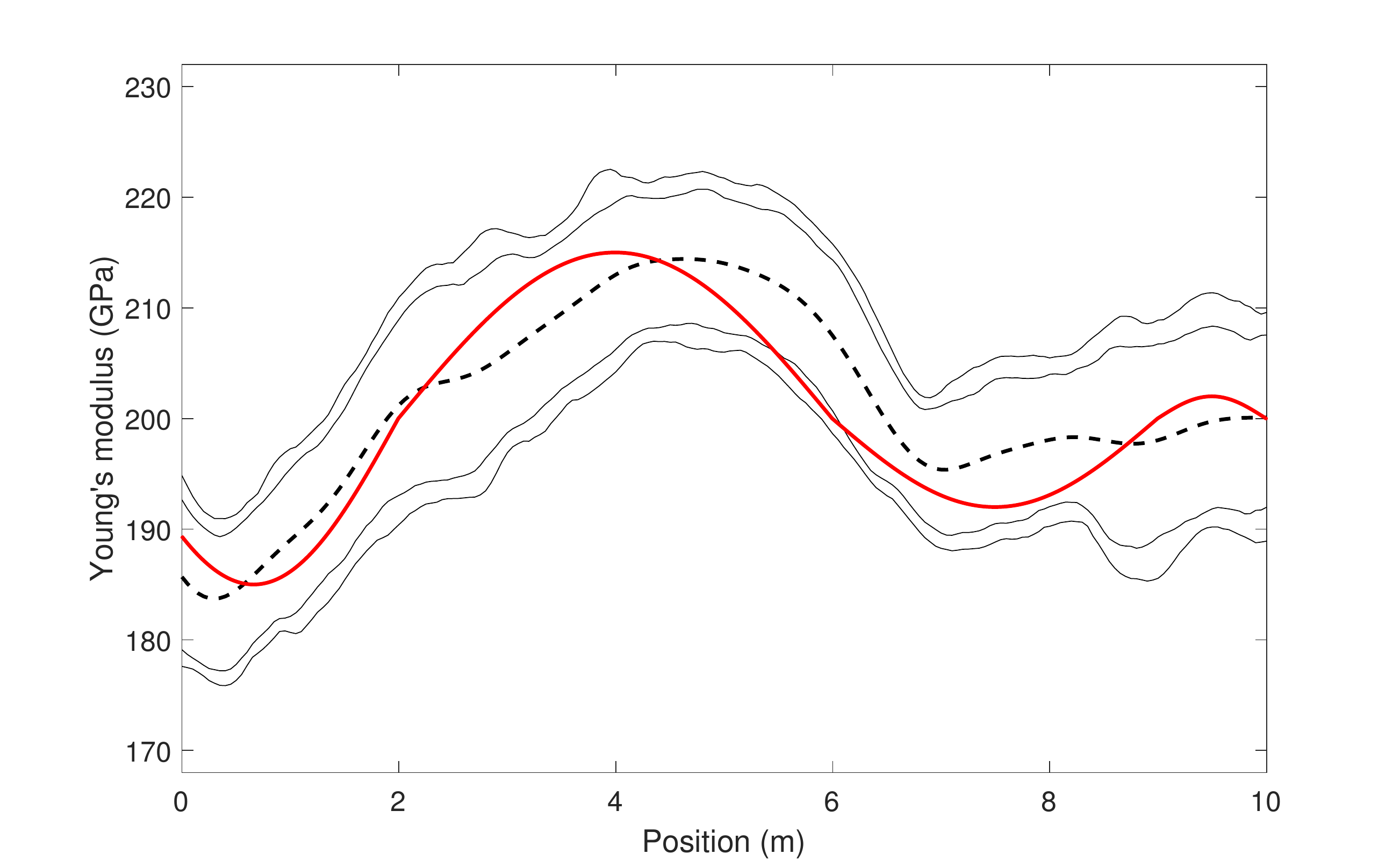}
			\caption{Continuous,  right observations.}
			\label{fig:i3}
		\end{subfigure}%
		\hfill
		\begin{subfigure}{.45\textwidth}
			\centering
			\includegraphics[height = 4.2cm,width=8cm]{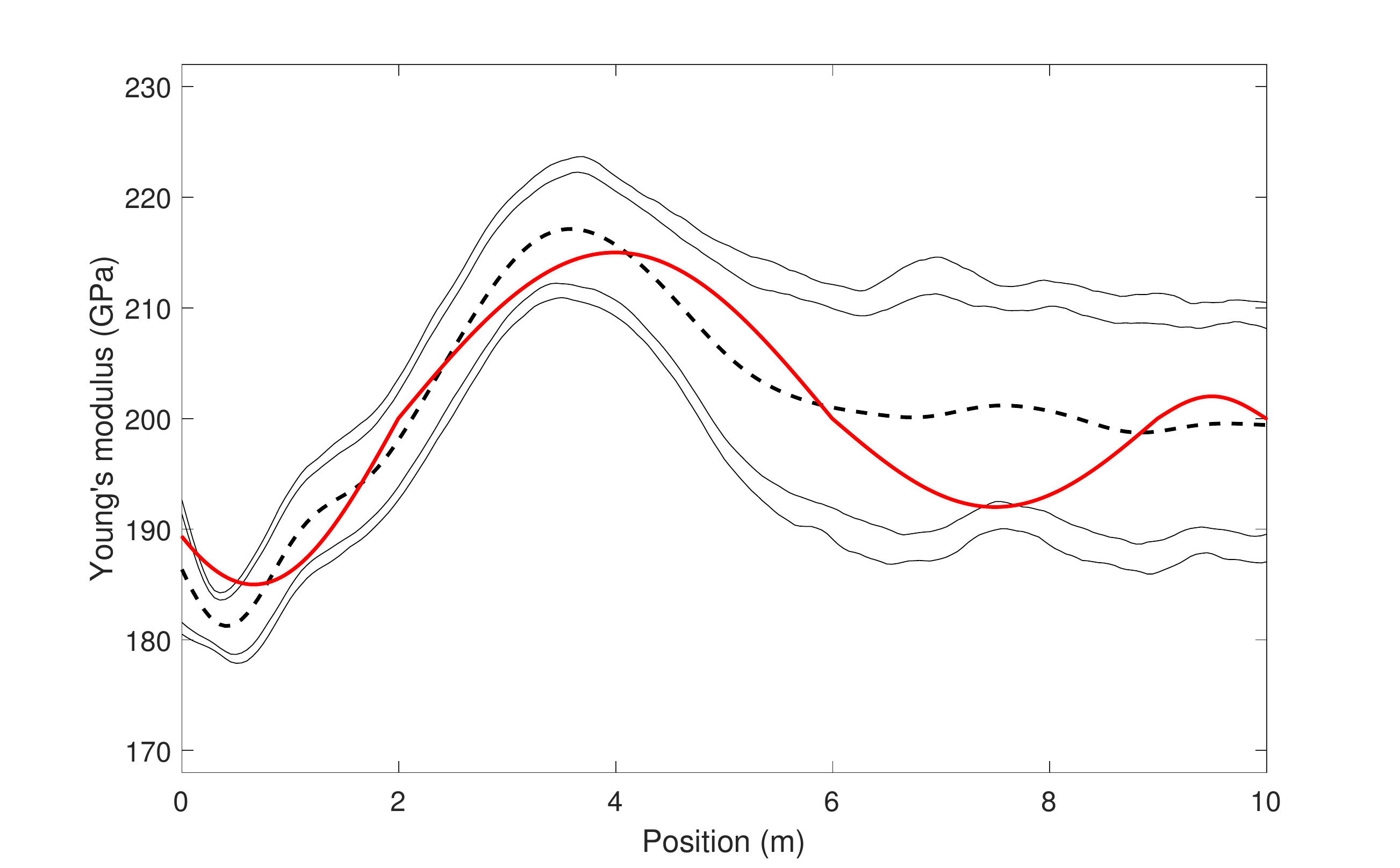}
			\caption{Continuous, left observations.}
			\label{fig:i4}
		\end{subfigure}%
		\caption{\label{figure:idealized} Idealized posterior $p_{u|y}(u)$, with mean (dashed black) and the 5, 10, 90, 95-percentiles (thin black), versus the true value (red).}
		\label{fig:i}
		\end{figure}

	\subsection{Euler-Maruyama Discretization of SDEs: a Signal Processing Application}\label{sec:SDE}
	Let $f:\R^d \to \R^d$ be globally Lipschitz continuous and consider the SDE 
	\begin{equation}\label{eq:SDE}
	d\z(t) = f(\z)\, dt + d u, \quad 0<t\le T, \quad \quad \z(0) = 0,
	\end{equation}
	where $u$ denotes $d$-dimensional Brownian motion.
	We aim to recover  $u$ from observations of the solution $\z$. We suppose that the observations $y = [y_1, \ldots, y_{m}]$ are given by 
	\begin{equation}
	y_i = \z(t_i) + \eta_i, \quad \quad i = 1, \ldots, m,
	\end{equation}
	where $\eta = [\eta_1, \ldots, \eta_{m}]$ is assumed to follow a centered Gaussian distribution with covariance $\Gamma$ and 
	$$0 < t_1< \cdots <t_{m} < T$$
	are given observation times.
	Following \cite{hairer2011signal},  we cast the problem in the setting of Section \ref{sec:fullproblem}.
	First note that the solution to the integral equation
	\begin{equation}\label{eq:FforSDE1}
	\z(t) = \int_0^t f \bigl(\z(s)\bigr)\, ds + u(t), \quad 0 \le t \le T, 
	\end{equation}\
	defines a map
	\begin{align}\label{eq:FforSDE2}
	\F: C([0,T], \R^d) &\to C([0,T], \R^d) \\
	u & \mapsto \z \;.
	\end{align}
	Thus we set the input and output space to be $\U = \Z = C([0,T], \R^d).$
	
	Next we define an observation operator
	\begin{align*}
	\O : C([0,T], \R^d) &\to \R^{m} \\
	\z &\mapsto [\z(s_1), \ldots, \z(s_{m})]
	\end{align*} 
	and set $\G = \O \circ \F.$
	We put as prior on $u$ the standard $d$-dimensional Wiener measure, that we denote $\mu_u.$ 
	Then the posterior distribution $\mu_{u|y}$ is given by Equation \eqref{eq:posterior}, which
	if $\Gamma = \gamma^2 I_{m}$ may be rewritten as
	\begin{equation}\label{eq:likelihoodexpanded}
	\frac{d\mu_{u|y}}{d\mu_u}(u) \propto \exp\left(  \frac{1}{2\gamma^2} \sum_{i=1}^{m} \big|y_i - \G(u)(s_i)\big|^2 \right).
	\end{equation} 
	Note that the likelihood does not involve evaluation of $\F(u)$ at times $t>s_{m},$ and hence changing the definition of $\F(u)(t)$ for $t>s_{m}$ does not change the posterior measure. Thus, we expect suitable discretizations of the forward model to refine finely only times up to the right-most observation.

	\subsubsection{Forward Discretization}
	For most nonlinear drifts $f$,  the integral equation \eqref{eq:FforSDE2} cannot be solved in closed form and a numerical method needs to be employed. A family of numerical solutions may be parameterized by the set
	$$\A:= \Bigl\{ a = (k, \theta): \, k \in \K \subset\{1,2, \ldots\},\, \theta =[t_1, \ldots, t_k] \in [0,T]^k \Bigr\}.$$
	Precisely, for $a= (k,\theta) \in \A$ we define an approximate, Euler-Maruyama solution map 
	\begin{align*}\label{eq:FforSDE2}
	\F^a: C([0,T], \R^d) &\to C([0,T], \R^d) \\
	u & \mapsto \z^a
	\end{align*}
	as follows. First, we reorder the $t_j$'s so that 
	$$0 =: t_0 \le t_1 \le \ldots \le t_k \le t_{k+1} := T.$$
	Then we define  $\z^{a}_j := \z^a(t_j)$  as $z_0^a =0$, and 
	\begin{equation}
	\z^a_{j+1} = \z_j^a + (t_{j+1} - t_{j})f(\z^a_j) + u(t_{j+1}) - u(t_j), \quad 1 \le j \le k.
	\end{equation}
	Finally, for $t \in (t_j, t_{j+1})$ we define $\z^a(t)$ by linear interpolation of $\z^a_j$ and   $\z^a_{j+1}.$
	
	Having defined the parameter space $\A$ we now describe a choice of prior distribution $\nu_a$ on $\A$ and the resulting combined prior $\nu_{u,a}$ on $(u,a) \in \U \times \A.$
	First we choose a prior $\nu_k$ on the number $k$ of grid-points. Given the number of grid points $k$, assuming no knowledge on appropriate discretizations for the SDE \eqref{eq:SDE} we put a uniform prior on grid locations $\theta = [t_1, \ldots, t_k].$ 
	
	\begin{remark}
		More information could be put into the prior. In particular it seems natural to impose that grids are finer at the beginning of the time interval.
	\end{remark}

	\subsubsection{Implementation Details and Numerical Results}\label{ssec:numericsSDE}
	For our numerical experiments we considered the SDE \eqref{eq:SDE} with $T=10$ and double-well drift 
	\begin{equation}
	f(t) = 10\, t \,(1-t^2)/(1+t^2).
	\end{equation}
	We generated synthetic observation data $y$ by solving \eqref{eq:SDE} on a very fine grid, and then perturbing the solution at uniformly distributed times $t_i = 0.2i,\, \, i = 1, \ldots, 24=m$ so that the last observation corresponds to time $t=4.8.$ The observation noise was taken to uncorrelated, $\Gamma = \gamma^2 I_{m},$ with $\gamma = 0.1.$ The motivation for choosing this example is that there is certain intuition as to where one would desire the discretization grid-points to concentrate. Indeed, since all the observations $t_i$ are in the interval $[0.2,4.8]$ it is clear from Equation \eqref{eq:likelihoodexpanded} that any discretization points $t_j\in (4.8,10]$ will not contribute to better approximate $\mu_{u|y}.$ In other words, those grid points would help in approximating $\F$ but not in approximating $\G=\O\circ \F$.
	
	We report our results for a small grid-size $k=24.$ Similar but less dramatic effect was seen for larger grid size. Precisely, we chose our set of admissible grids to be given by
	$$[t_0, \ldots, t_{25}] : 0.01= t_0 \le  t_1 \le \ldots \le t_{25}  = 10 .$$
	For implementation purposes,  elements in the space $\U=C([0,T],\R)$ were represented as vectors in $\R^{1000}$ containing their values on a uniform grid of step-size $0.1$. We run these algorithms with parameter choices $N= 10^5,$ $\beta = 0.1,$ $\zeta=0.5.$

The experiments show a successful reconstruction of the SDE path.  Moreover, the grids concentrate in $[0,4.8]$ in agreement with our intuition and the uncertainty quantification is satisfactory. In contrast, we see that when using the same number of grid points but on a uniform grid the Euler-Maruyama scheme is unstable, leading to a collapse of the MCMC algorithm. Then, the posterior constructed with a uniform grid completely fails at reconstructing the SDE path, and the uncertainty quantification is overoptimistic due to poor mixing of the chain.


			\begin{figure}
		\centering
		\begin{subfigure}{.50\textwidth}
			\centering
			\includegraphics[width=1\linewidth]{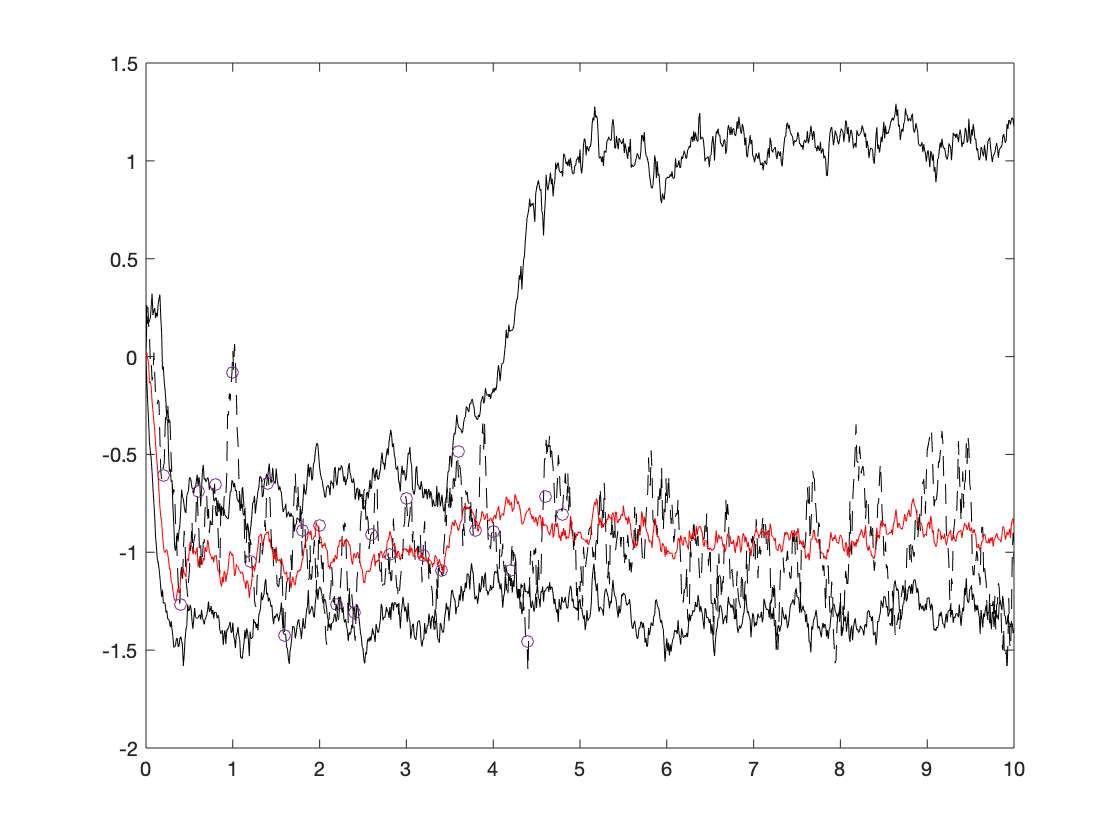}
			\caption{Data-driven discretization}.
			\label{fig:cont_rgest}
		\end{subfigure}
		\hfill
		\begin{subfigure}{.47\textwidth}
			\centering
			\includegraphics[width=1\linewidth]{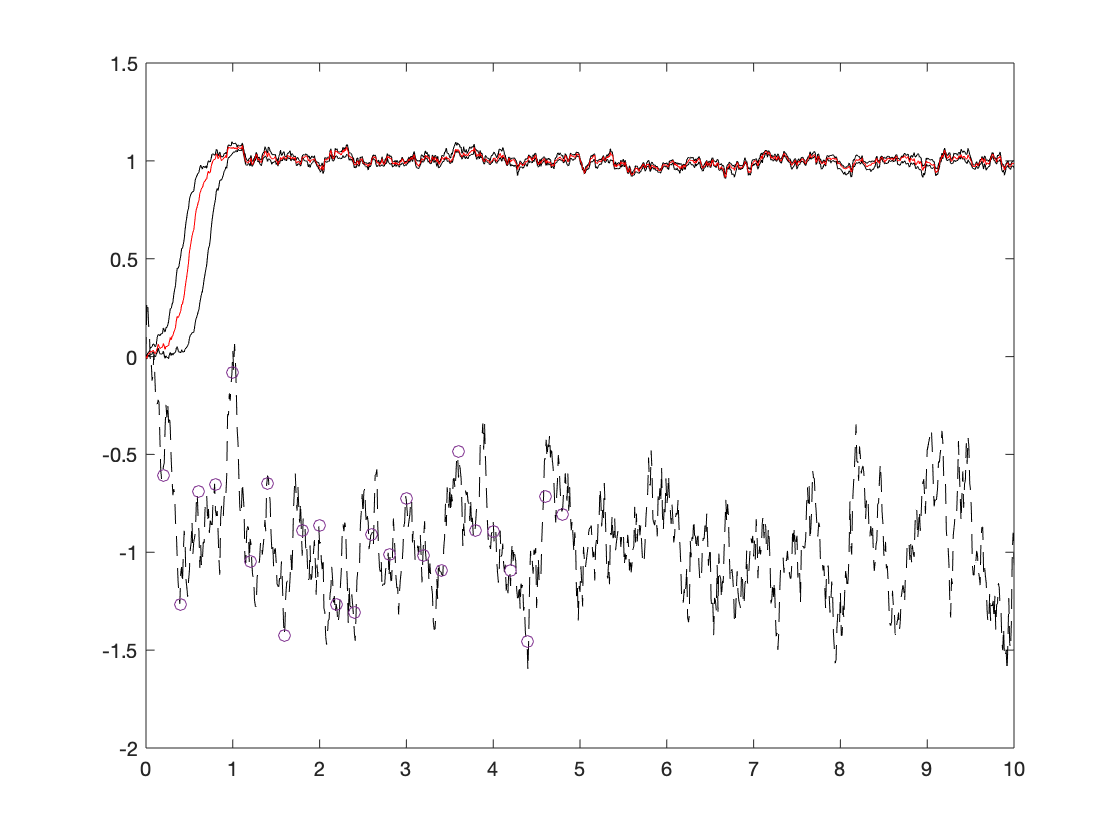}
			\caption{Uniform discretization.}
			\label{fig:cont_rgestf}
		\end{subfigure}
		\caption{Recovered SDE trajectory in the time-interval $t \in [0,10]$. The true trajectory is shown in dashed black line. The posterior median is shown in red, and 5 and 95-percentiles are shown in black. The small circles denote the locations of the observations.} 
		\label{fig:r}
	\end{figure}

	\subsection{Finite Element Discretization: Source Detection}\label{ssec:FEM}
	Consider the boundary value problem
	\begin{equation}
	\label{eq:sourcedet}
	\left\lbrace\,
	\begin{array}{@{}r@{}l@{\quad}l@{}}
		-\Delta z(x)&=\delta(x-u), &x\in D,\\
		\z(x)&=0, &x\in\partial D,
	\end{array}
	\right.
	\end{equation}	
	where $D=(0,1)\times(0,1)\subset\R^2$ is the unit square and $\delta$ is the Dirac function at the origin. We aim to recover the source location $u$ from sparse observations
	\begin{equation}
	y_i=z(s_i)+\eta_i, \quad \quad i=1,\dots, m,
	\end{equation} 
	where $\eta=[\eta_1,\dots,\eta_{m}]$ follows a centered Gaussian distribution with covariance $\Gamma$ and $s_1,\dots,s_{m}\in D\setminus \{u\}$ are observation locations. To cast the problem in the setting of Section \ref{sec:fullproblem}, we let  $\F$ be given by Green's function for the Laplacian on the unit square (which does not admit an analytical formula but can be computed e.g. via series expansions \cite{melnikov2006computability}) and $\O$ be defined by point-wise evaluation at the observation locations. 
%
The prior on $u$ is the uniform distribution in the unit square $D$, which we denote $p_u(u)$. Since $\U=D$ is finite dimensional, the posterior has Lebesgue density given by Equation \eqref{eq:posteriorlebesgue}.
We find this problem to be a good test, as there is a clear understanding that the data-driven mesh should concentrate around the source.
	\subsubsection{Forward Discretization}
	To solve equation \eqref{eq:sourcedet} numerically we employ the finite element method. The use of uniform grid is here wasteful, as the mesh should ideally concentrate around the unknown source $u.$
	
	We will use grids obtained as the Delauny triangulation of central Voronoi tessellations $\{V_i\}_{i=1}^k$ and generators $\{x_i\}_{i=1}^k$, where each $x_i\in D$ and $V_i\subset D$. This can be calculated as the solution of the optimization problem, parameterized by a probability density $\rho$ on $D$:
	\begin{equation}
		\min_{\{x_i\}\subset D, \{V_i\}}\sum_{i=1}^k\int_{V_i}\rho(x)\|x-x_i\|^2\md x
	\end{equation}
	subject to the constraint that $\{V_i\}_{i=1}^k$ is a tessellation of $D$. One can refer to \cite{du2002} for more details.
	For a fixed density $\rho$ and integer $k$ we denote the optimal grid points by $\{x_{\rho,i}\}_{i=1}^k$. Then the approximated solution map is defined as
	\begin{align}
	\label{eq:approxdet}
	\F^a:D&\rightarrow H_0^1(D)\\
	u&\mapsto z^a
	\end{align}
	where $z^a$ is given by the finite element solution of equation \eqref{eq:sourcedet} with respect to (the Delauny triangulation of) the grid points $\{x_{\rho,i}\}_{i=1}^k$. Details on the creation of grid for prescribed parameters $\rho$ and $k$ will be discussed below.
	
	In the spirit of adapting the grid to favor the ones maximizing the model evidence, we constrain $\rho$ to belong to a family of parametric densities $\Pi=\{\rho(x;\theta)|\theta\in\R^P\}$ where $\rho(x;\theta)=\mbox{Beta}(\alpha_1,\beta_1)\times\mbox{Beta}(\alpha_2,\beta_2)$ is the product measure of two Beta distributions. Therefore in this case $\theta=(\alpha_1,\beta_1,\alpha_2,\beta_2)$ and $P=4$. Each pair $(k, \theta)$ describes a member in the discretization family $\A$, where $k$ controls the number of grid points, while $\theta$ controls how these grid points are distributed in the spatial dimension.
	
	
	\subsubsection{Implementation Details and Numerical Results}
	We solved equation \eqref{eq:sourcedet} on a fine grid $k=2000$ with the true point source $u^*=(0.85, 0.85)$. The observation locations were $\{s_1,\dots,s_{25}\}=\{0.5,0.6,0.7,0.8,0.9\}\times\{0.5,0.6,0.7,0.8,0.9\}$. Observation noise was uncorrelated with $\Gamma=\gamma^2I_{25}$, $\gamma=0.05$. 
	
	In this example the prior of $u$ is the uniform distribution on $D=(0,1)\times(0,1)$ and $u$ is initialized at $(0.2,0.2)$. The parameters $\theta$ are also set to have a uniform prior $\theta=(\alpha_1,\beta_1,\alpha_2,\beta_2)\sim\textit{Uniform}\bigl([1,10]^4\bigr)$. We initialize $(\alpha_1,\beta_1,\alpha_2,\beta_2)=(1,1,1,1)$, which corresponds to (near) uniform grid in $D$. For simplicity, in this experiment we set a point mass prior on $k$, with $k=100.$ We run the algorithm with $N= 10^4.$
	
	We compare our algorithm to the traditional method where we fix a uniform grid in $D$ and run the MCMC algorithm only on $u$. We found out that with the same number of grid points, our data-driven approach gives a posterior distribution $q_{u|y}$ that is more concentrated around the true location of the point source, as shown in Figure \ref{fig:fixedmeshdet} and \ref{fig:movingmeshdet}. Also, Figure \ref{fig:finalmesh} shows that the adaptive discretization is concentrated at the top right corner of the region, where the hidden point source $u^*$ is located.

	Next we show that data-driven discretizations of the forward map can be employed to provide improved uncertainty quantification of the PDE solution, and not only to better reconstruct the unknown input. To illustrate this, we approximate the pushforward distribution $\F_\sharp(q_{u|y})$ in three different ways, as shown in Figure \ref{fig:sourcedet}.  Let $(u^{(n)},a^{(n)})$ denote the output of the Gibbs sampling algorithm at iteration $n$. We first consider the traditional method where the grid is fixed and uniform, that is, $a^{(n)}=a$ is fixed. Then the pushforward distribution can be well-approximated by $\{\F^a(u^{(n)})\}_{n=1}^{N}$, for $N$ large enough. We then consider the same setting except that $\F^a$ is replaced by a forward map $\F$ computed in a fine grid $k=2000$ and the pushforward is approximated by $\{\F(u^{(n)})\}_{n=1}^{N}$. Finally we consider a data-driven setting stemming from our algorithm, where the pushforward distribution is approximated by $\{\F^{a^{(n)}}(u^{(n)})\}_{n=1}^{N}$. We see that our algorithm reconstructs well the solution to the PDE, with a more accurate mean and a smaller variance.

	\begin{figure}
	\centering
	\begin{subfigure}{.33\textwidth}
		\centering
		\includegraphics[width=.99\linewidth]{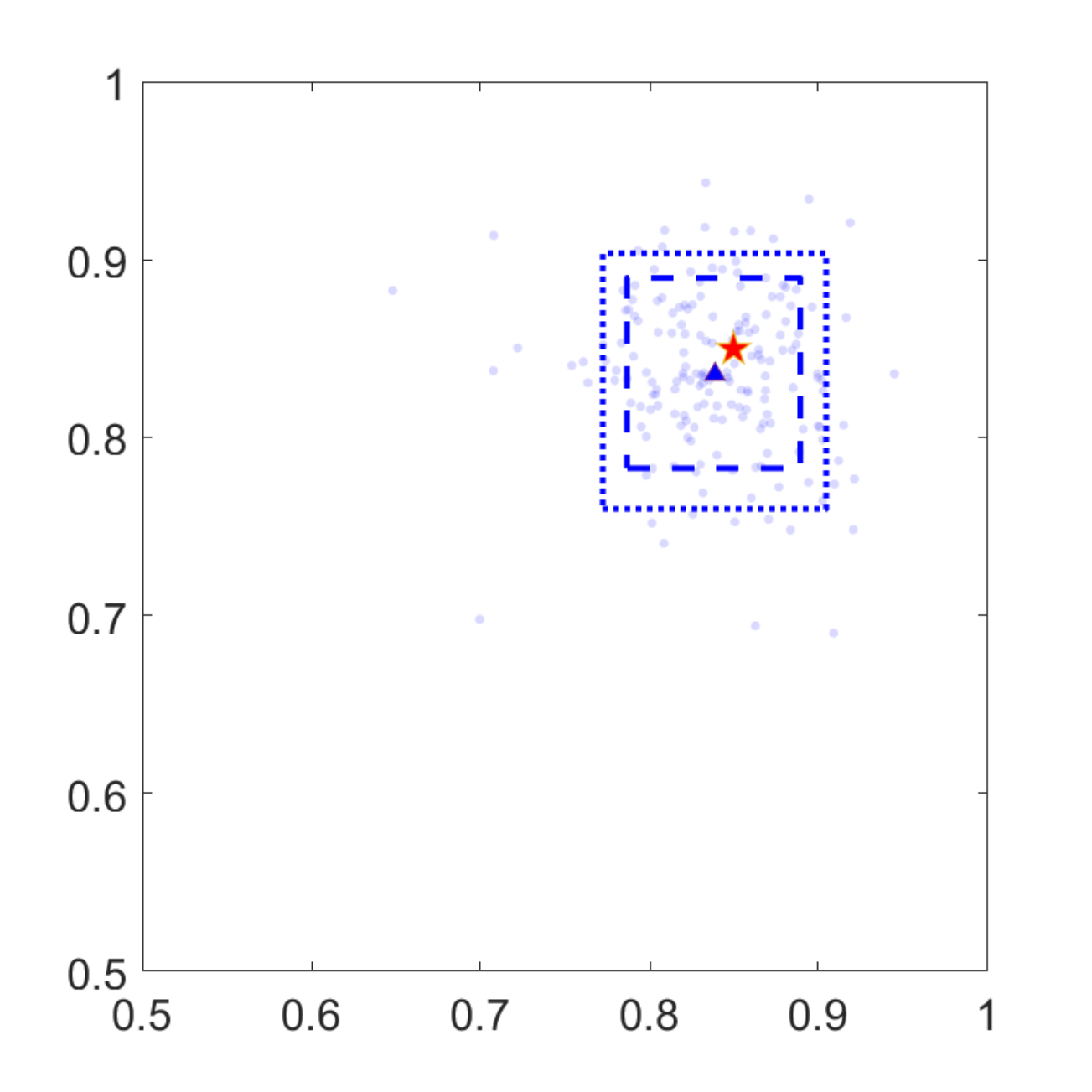}
		\caption{Fixed grid.}
		\label{fig:fixedmeshdet}
	\end{subfigure}%
	\hfill
	\begin{subfigure}{.33\textwidth}
		\centering
		\includegraphics[width=1\linewidth]{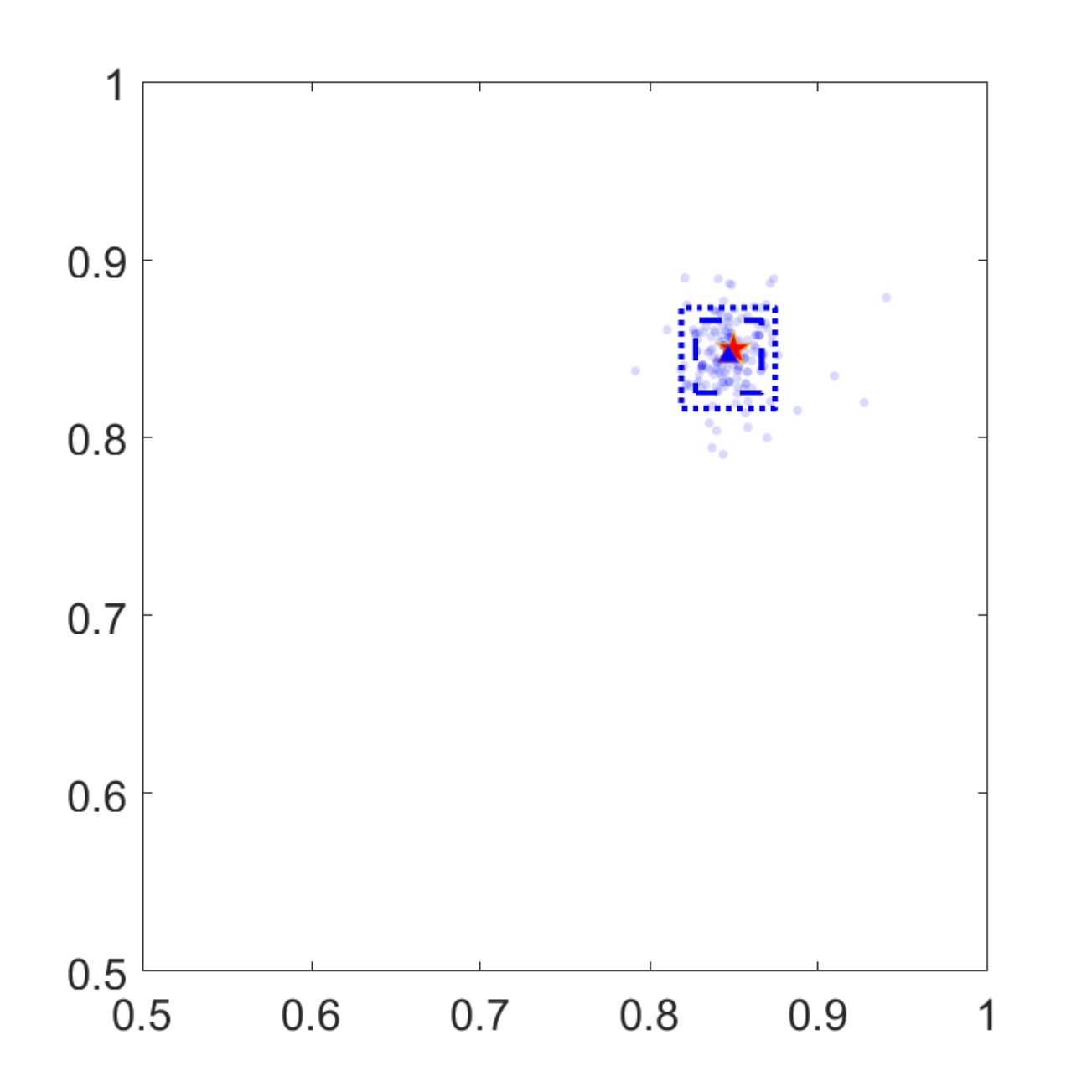}
		\caption{Data-driven grid.}
		\label{fig:movingmeshdet}
	\end{subfigure}
	\hfill
	\begin{subfigure}{.33\textwidth}
		\centering
		\includegraphics[width=1\linewidth]{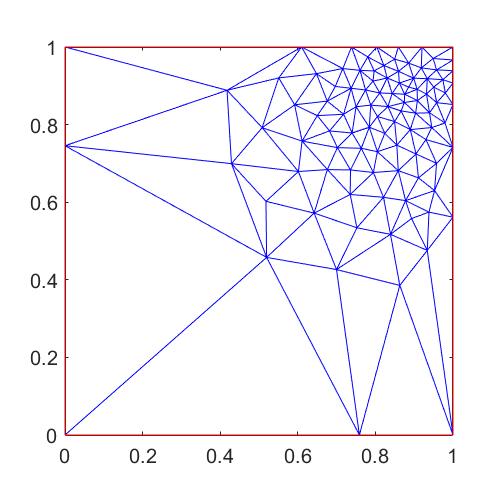}
		\caption{Grid generated in the last iteration.}
		\label{fig:finalmesh}
	\end{subfigure}
	\label{fig:resdet}
	\caption{Figures \ref{fig:fixedmeshdet} and \ref{fig:movingmeshdet} show the posterior distribution $q(u|y)$, where the grid is fixed and uniform in \ref{fig:fixedmeshdet}, and data-driven in \ref{fig:movingmeshdet}. Red star indicates the true location of the source, blue dots are random samples from the posterior, blue triangle is the posterior mean, and dash (resp. dotted) lines correspond to the 90\% (resp. 95\%) coordinate-wise  credible regions. Figure \ref{fig:finalmesh} shows the grid generated in the last iteration of the MCMC update.}
	\end{figure}
	\begin{figure}
		\centering
		\includegraphics[width=.3\textwidth]{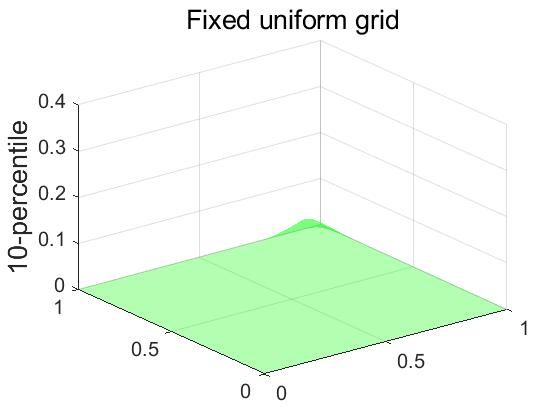}
		\includegraphics[width=.3\textwidth]{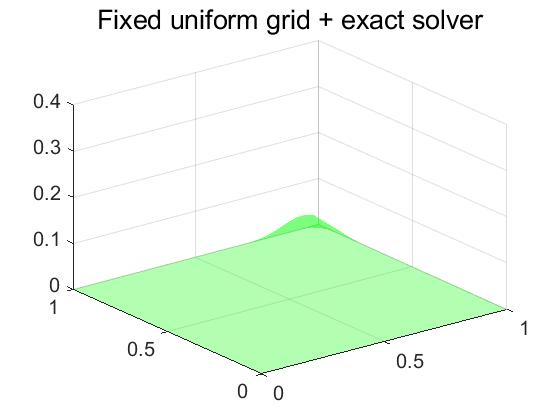}
		\includegraphics[width=.3\textwidth]{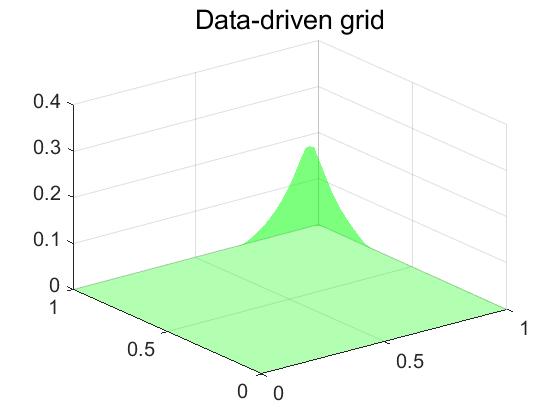}
		\\
		\includegraphics[width=.3\textwidth]{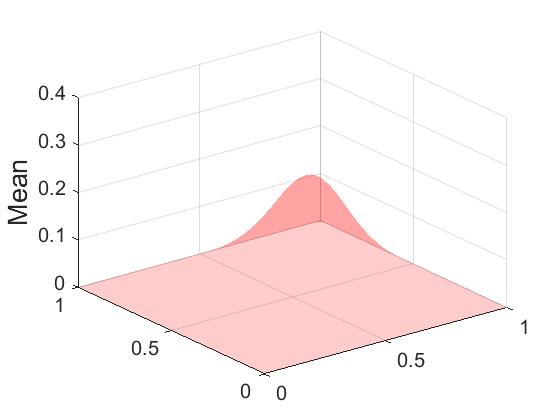}
		\includegraphics[width=.3\textwidth]{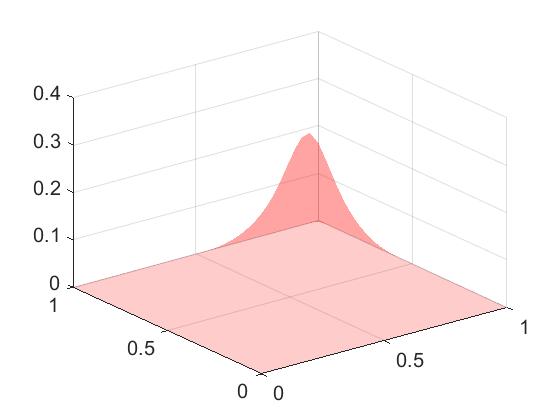}
		\includegraphics[width=.3\textwidth]{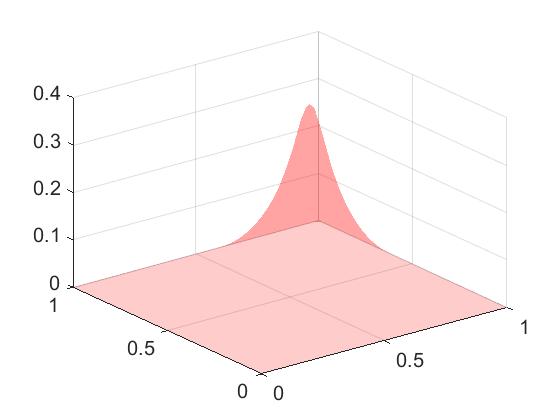}
		\\
		\includegraphics[width=.3\textwidth]{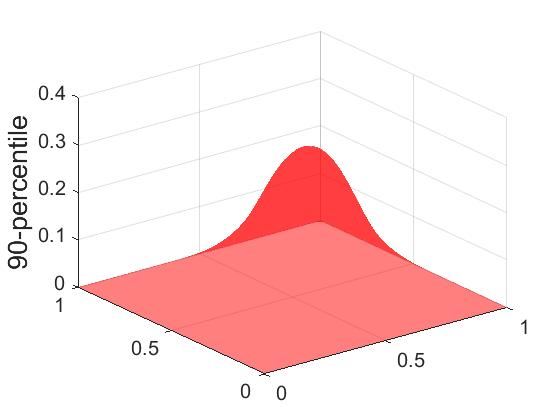}
		\includegraphics[width=.3\textwidth]{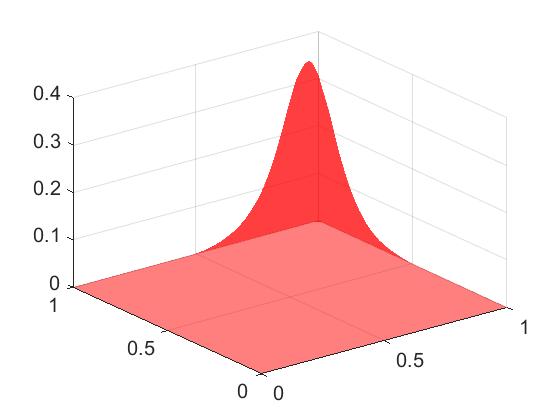}
		\includegraphics[width=.3\textwidth]{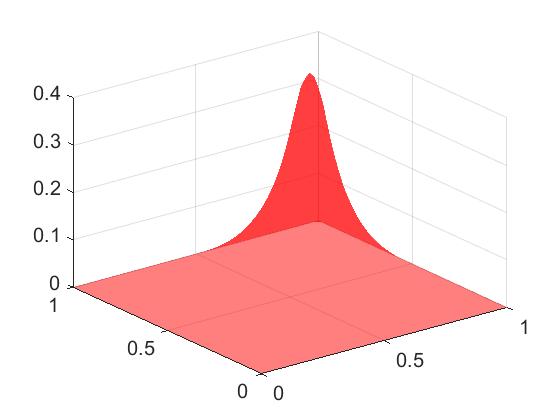}

		\label{fig:pushforwarddet}
	\caption {The mean, 10 and 90-percentile of the pushforward distribution $\F_\sharp(q_{u|y})$ under three different settings: (1) Both the posterior $q_{u|y}$ and its pushforward $\F_\sharp(q_{u|y})$ are computed on a fixed and uniform grid; (2) The posterior $q_{u|y}$ is computed on a fixed and uniform grid, and its pushforward $\F_\sharp(q_{u|y})$ is calculated using a (nearly) exact solver; (3) Both the posterior and its pushforward are computed on a data-driven grid.}
	\label{fig:sourcedet}
	\end{figure}	
	
	\section{Conclusions and Open Directions}\label{sec:conclusions}
	\begin{itemize}
	\item We have shown that, in a variety of inverse problems, the observations contain useful information to guide the discretization of the forward model, allowing a better reconstruction of the unknown than using uniform grids with the same number of degrees of freedom.  Despite these results being promising, it is important to note that updating the discretization parameters may be costly in itself, and may result in slower mixing of the MCMC methods. For this reason, we envision that the proposed approach may have more potential when the computational solution of the inverse problem is very sensitive to the discretization of the forward map and discretizing it is expensive. We also believe that density-based discretizations may help in alleviating the cost of discretization learning.
		\item  An interesting avenue of research stemming from this work is the development of prior discretization models  that are informed by numerical analysis of the forward map $\F$, while recognizing the uncertainty in the best discretization of the forward model $\G.$ Moreover, more sophisticated prior models beyond the product structure considered here should be investigated.
		\item Topics for further research include the development of new local proposals and sampling algorithms for grid-based discretizations, and the numerical implementation of the approach in computationally demanding inverse problems beyond the proof-of-concept ones considered here.
	\end{itemize}	
	
	  \section*{Acknowledgement}
  The work of NGT and DSA was supported by the NSF Grant DMS-1912818/1912802. The work of DB and YMM was supported by NSF Grant DMS-1723011.
    
    \newpage
    
	\bibliographystyle{plain}
	\bibliography{isbib}
	
\newpage

	\appendix
	\section{Algorithm Pseudo-Code}
			\begin{algorithm}
		\caption{Metropolis within-Gibbs}
		\label{algorithm:gridsampling}
		\begin{algorithmic}
			\State {\bf Input parameters}: $\beta$ (pCN  step-size), $\zeta$ (probability of location moves),  $N$ (sample size).
			\State Choose $(u^{(1)},a^{(1)})\in \U \times \A$.
			\For{$n = 1:N$} 
			\State {\bf Stage I} Do a pCN move to update $u$ given $a,y:$
			\State
			\begin{enumerate}[i)]
			\vspace{-0.5cm}
				\item  Propose $\tilde{u}^{(n)} = \sqrt{1-\beta^2} \, u^{(n)} + \beta v^{(n)}, \quad \quad v^{(n)}\sim \mu_u.$
				\item Set $u^{(n+1)} = \tilde{u}^{(n)}$ with probability 
				$$a(u^{(n)}, \tilde{u}^{(n)}) = \min\left \{1, \exp\Bigl(\Psi(u^{(n)},a^{(n)};y) - \Psi(\tilde{u}^{(n)},a^{(n)};y)\Bigr) \right\}.$$
				\item Set $u^{(n+1)} = u^{(n)}$ otherwise.
			\end{enumerate}
			\State {\bf Stage II} Update $a = (k,\theta)$ given $u$ and $y.$	
			\State {\bf Stage IIa} With probability $\zeta,$ update $\theta$ given $u,y$ with a grid re-location step:
			\begin{enumerate}[i)]
				\item Propose $\tilde{a}^{(n)}$ by picking one of the $k$ interior grid points of $a^{(n)}$ uniformly at random, and replacing it by a uniform draw in $D.$
				\item Set $a^{(n+1)} = \tilde{a}^{(n)}$ with probability 
				$$\alpha(a^{(n)}, \tilde{a}^{(n)}) = \min\left \{1, \exp\Bigl(\Psi(u^{(n+1)},{a^{(n)}};y) - \Psi(u^{(n+1)},\tilde{a}^{(n)};y)\Bigr)\right\}.$$
				\item Set $a^{(n+1)} = a^{(n)}$ otherwise.
			\end{enumerate}
			\State {\bf Stage IIb} Otherwise, (with probability $1-\zeta$) update $k$ with a birth/death step:
			\begin{enumerate}[i)]
				\item Propose a new number $\tilde{k}^{(n)}$ of grid-points. 
				\item If $\tilde{k}^{(n)} \le k^{(n)}$ remove uniformly chosen grid-points.
				\item If $\tilde{k}^{(n)} > k^{(n)}$ draw required number of new grid points uniformly at random in $D.$
				\item Set $a^{(n+1)} = \tilde{a}^{(n)}$ with probability 
				$$\alpha(a^{(n)}, \tilde{a}^{(n)}) = \min\left \{1,  \frac{\nu_k(\tilde{k}^{(n)})}{\nu_k(k^{(n)})}  \exp\Bigl(\Psi(u^{(n+1)},{a^{(n)}};y) - \Psi(u^{(n+1)},\tilde{a}^{(n)};y)\Bigr) \right\}.$$
				\item Set $a^{(n+1)} = a^{(n)}$ otherwise.
			\end{enumerate}
			\EndFor
		\end{algorithmic}
	\end{algorithm}
\newpage
\section{Additional Results for Section 5.1}
		\begin{table}[!htb]
		\begin{tabular}{|l|l|l|l|l|l|l|l|l|l|l|}
			\hline
			& (0, 1) & (1, 2) & (2, 3) & (3, 4) & (4, 5) & (5, 6) & (6, 7) & (7, 8) & (8, 9) & (9, 10) \\ \hline
			2-3   & 0.0015 & 0.0047 & 0.0116 & 0.0094 & 0.0164 & 0      & 0      & 0      & 0      & 0       \\ \hline
			4-5   & 0.0791 & 0.1089 & 0.2003 & 0.1721 & 0.1875 & 0.0260 & 0.0967 & 0.0363 & 0.0554 & 0       \\ \hline
			6-7   & 0.3345 & \textbf{0.3940} & \textbf{0.4002} & \textbf{0.4017} & \textbf{0.3956} & 0.2309 & 0.2817 & 0.1789 & 0.3529 & 0.0578  \\ \hline
			8-9   & \textbf{0.3956} & 0.3412 & 0.2785 & 0.2967 & 0.2907 & \textbf{0.3548} & \textbf{0.3011} & \textbf{0.4835} & \textbf{0.3944} & 0.3808  \\ \hline
			10-11 & 0.1548 & 0.1268 & 0.0920 & 0.0983 & 0.0924 & 0.2579 & 0.2612 & 0.2565 & 0.1635 & \textbf{0.3886} \\ \hline
			12-13 & 0.0289 & 0.0230 & 0.0153 & 0.0199 & 0.0159 & 0.0988 & 0.0550 & 0.0425 & 0.0312 & 0.1524  \\ \hline
			14-15 & 0.0052 & 0.0012 & 0.0016 & 0.0019 & 0.0016 & 0.0255 & 0.0044 & 0.0020 & 0.0027 & 0.0195  \\ \hline
			16-17 & 0.0004 & 0.0003 & 0.0004 & 0      & 0      & 0.0051 & 0      & 0.0002 & 0      & 0.0009  \\ \hline
			18-19 & 0      & 0      & 0      & 0      & 0      & 0.0009 & 0      & 0      & 0      & 0       \\ \hline
			
		\end{tabular}
		\caption{\label{tb:1} Distribution of grid points when observations are concentrated on the right, in the piecewise-constant Young' modulus case. Element on $i^{th}$ row and $j^{th}$ column represents the posterior probability of having $i$ grid points in the subinterval $j$.}
		\label{tb:rgriddistri}
	\end{table}
	
	\begin{table}[!htb]
		\begin{tabular}{|l|l|l|l|l|l|l|l|l|l|l|}
			\hline
			& (0, 1) & (1, 2) & (2, 3) & (3, 4) & (4, 5) & (5, 6) & (6, 7) & (7, 8) & (8, 9) & (9, 10) \\ \hline
			2-3   & 0      & 0      & 0      & 0      & 0      & 0.1712 & 0.1481 & 0.1435 & 0.1259 & 0.0605  \\ \hline
			4-5   & 0      & 0.0496 & 0.0076 & 0      & 0.0325 & \textbf{0.3420} & \textbf{0.3139} & \textbf{0.3183} & \textbf{0.3162} & 0.2362  \\ \hline
			6-7   & 0      & 0.2488 & 0.1263 & 0.0257 & 0.2238 & 0.2860 & 0.2995 & 0.3133 & 0.3087 & \textbf{0.3360}  \\ \hline
			8-9   & 0.0497 & \textbf{0.3785} & 0.3236 & 0.2174 & \textbf{0.4048} & 0.1341 & 0.1542 & 0.1508 & 0.1615 & 0.2316  \\ \hline
			10-11 & 0.2055 & 0.2390 & \textbf{0.3357} & \textbf{0.4114} & 0.2623 & 0.0383 & 0.0540 & 0.0483 & 0.0554 & 0.0954  \\ \hline
			12-13 & \textbf{0.3384} & 0.0717 & 0.1573 & 0.2492 & 0.0705 & 0.0062 & 0.0118 & 0.0097 & 0.0131 & 0.0274  \\ \hline
			14-15 & 0.2516 & 0.0111 & 0.0424 & 0.0789 & 0.0061 & 0.0012 & 0.0016 & 0.0009 & 0.0017 & 0.0053  \\ \hline
			16-17 & 0.1163 & 0.0013 & 0.0070 & 0.0155 & 0      & 0      & 0      & 0      & 0      & 0.0009  \\ \hline
			18-19 & 0.0321 & 0      & 0.0009 & 0.0018 & 0      & 0      & 0      & 0      & 0      & 0       \\ \hline
			20-21 & 0.0053 & 0      & 0      & 0      & 0      & 0      & 0      & 0      & 0      & 0       \\ \hline
			22-23 & 0.0007 & 0      & 0      & 0      & 0      & 0      & 0      & 0      & 0      & 0       \\ \hline
		\end{tabular}
		\caption{\label{tb:2} Distribution of grid points when observations are concentrated on the left, in the piecewise-constant Young' modulus case. Element on $i^{th}$ row and $j^{th}$ column represents the posterior probability of having $i$ grid points in the subinterval $j$.}
		\label{tb:lgriddistri}
	\end{table}
	
	\begin{figure}[!htb]
		\centering
		\begin{subfigure}{.49\textwidth}
			\centering
			\includegraphics[width=1\linewidth]{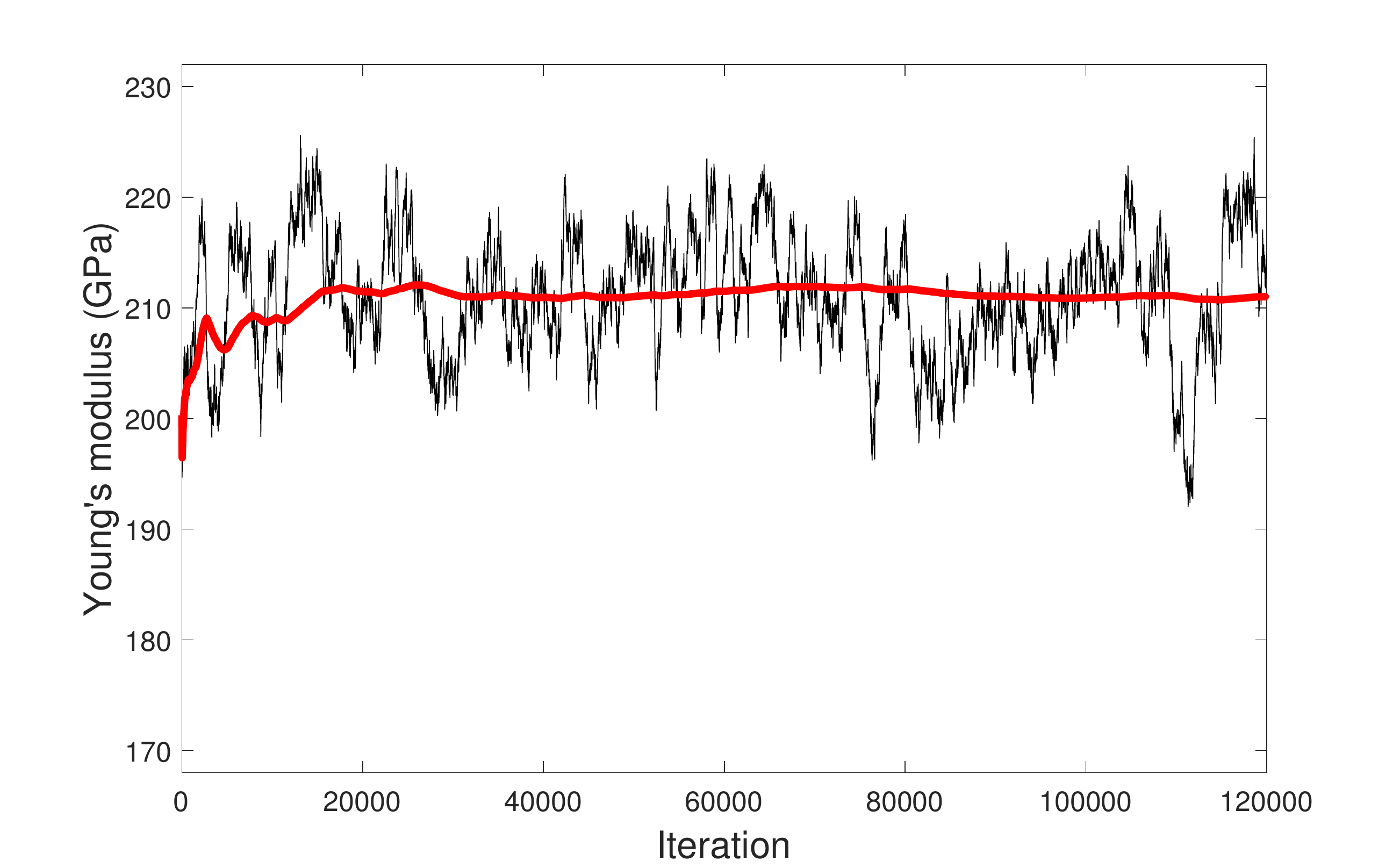}
			\caption{Continuous, right observations, $x=4$.}
			\label{fig:runningmeanr}
		\end{subfigure}
		\hfill
		\begin{subfigure}{.49\textwidth}
			\centering
			\includegraphics[width=1\linewidth]{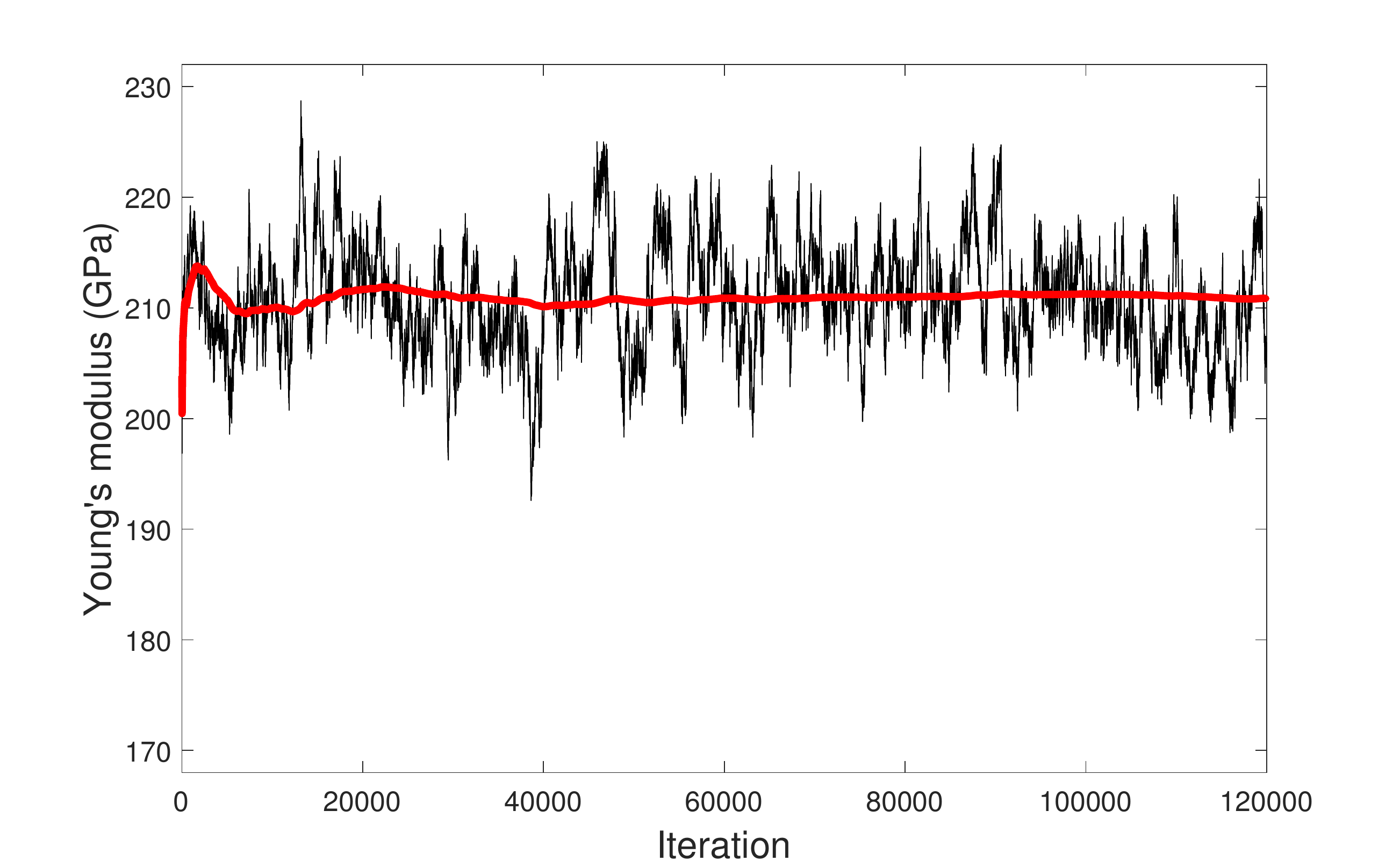}
			\caption{Continuous, left observations, $x=4$.}
			\label{fig:runningmeanl}
		\end{subfigure}%
		\vskip\baselineskip
		\begin{subfigure}{.49\textwidth}
			\centering
			\includegraphics[width=1\linewidth]{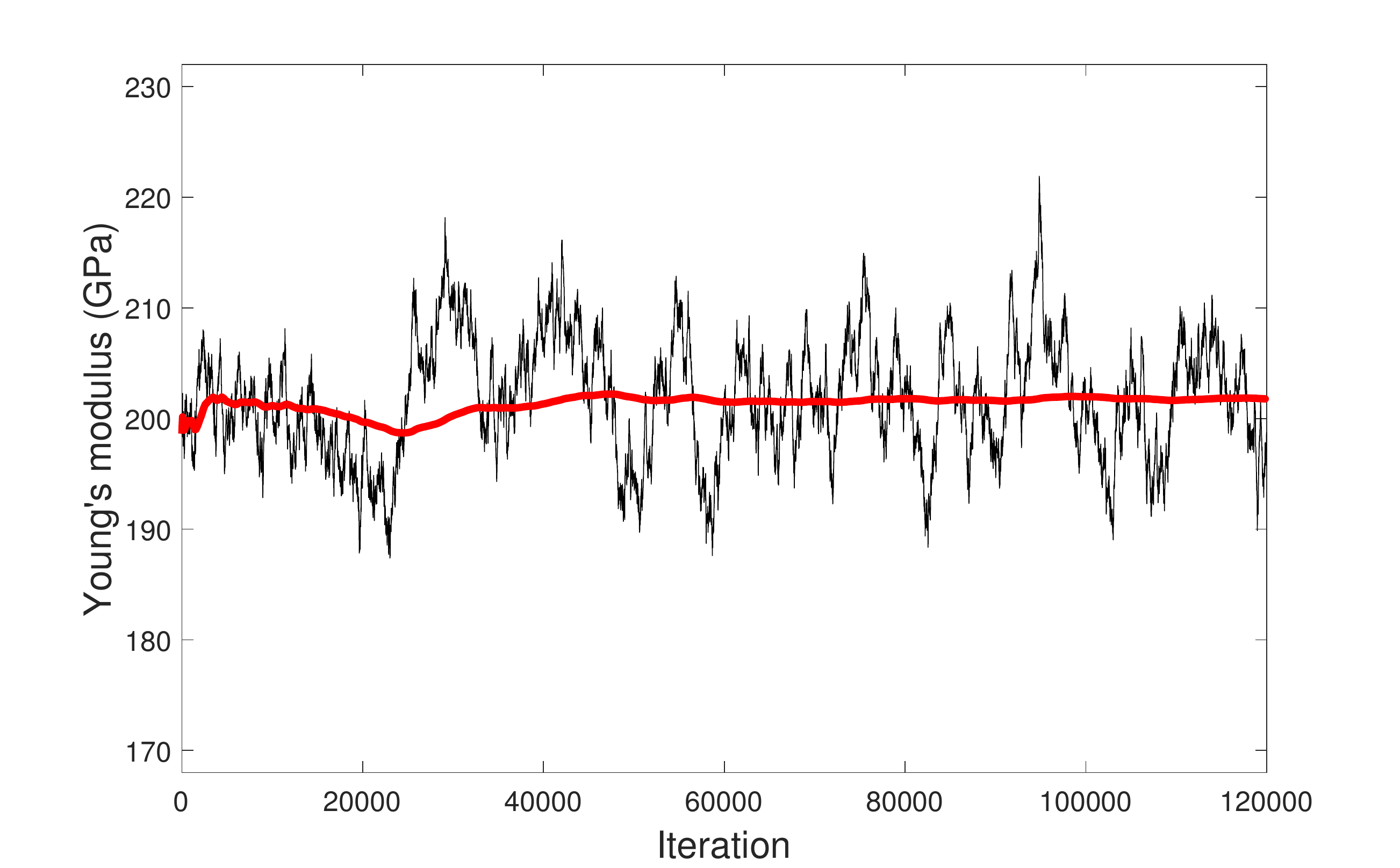}
			\caption{Continuous, right observations, $x=8$.}
			\label{fig:runningmeanr}
		\end{subfigure}
		\hfill
		\begin{subfigure}{.49\textwidth}
			\centering
			\includegraphics[width=1\linewidth]{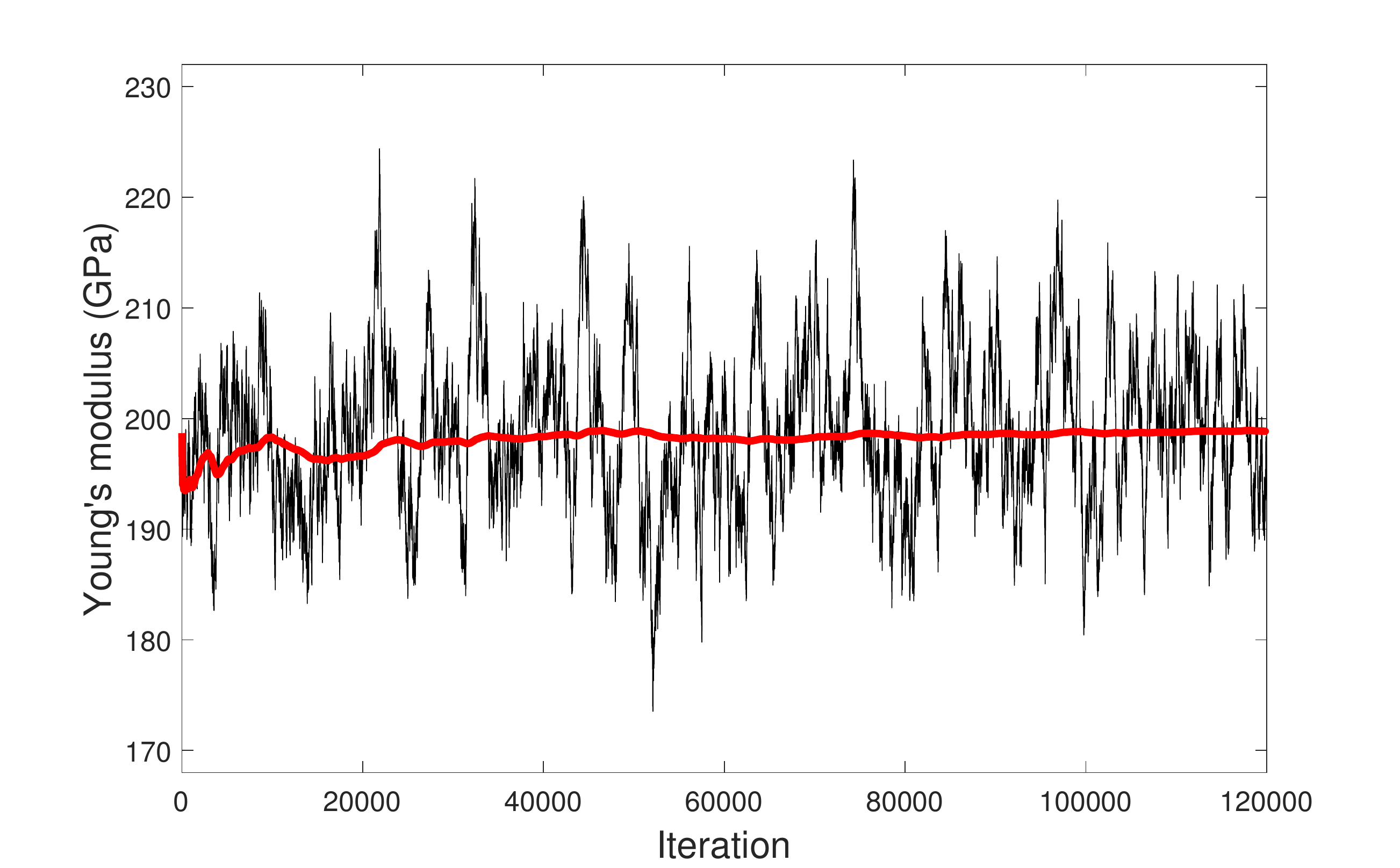}
			\caption{Continuous, left observations, $x=8$.}
			\label{fig:runningmeanl}
		\end{subfigure}%
		\caption{History of MCMC samples (black line) and running sample averages (red line) of continuous Young's modulus $u(x)$, at fixed locations $x=4$ and $x=8$ repectively, suggesting stationarity of the Markov chain.}
		\label{runningsample}
	\end{figure}
	
	\begin{table}[!htb]
		\centering
		\begin{tabular}{|l|l|l|}
			\hline
			& right obs. & left obs. \\ \hline
			$u$ & 0.2728     & 0.4590    \\ \hline
			$a$       & 0.1953    & 0.2314    \\ \hline
		\end{tabular}
	\caption{Averaged acceptance probability of $u$ and $a$ respectively, in the continuous Young' modulus case.}
	\label{accepttable}
	\end{table}

\end{document}